\documentclass{lmcs} 


\usepackage{hyperref}
\usepackage{amssymb}
\usepackage{amsmath}
\usepackage{lipsum}
\usepackage{bbm}
\usepackage{mathrsfs}
\usepackage{cmll}
\usepackage{stmaryrd}
\usepackage{faktor}
\usepackage{ebproof}
\usepackage{tikz-cd}
\usepackage{csquotes}
\theoremstyle{plain} %\crefname{satz}{Satz}{S\"atze}
\newcommand{\id}{\mathrm{id}}
\newcommand{\set}[1]{\{#1\}}
\newcommand{\banana}[1]{\llparenthesis#1\rrparenthesis}
\newcommand{\comp}{\circ}

\renewcommand{\implies}{\Rightarrow}
\newcommand{\op}{\mathit{op}}

\newcommand{\Type}{\mathrm{Type}}
\newcommand{\U}{\mathcal{U}}
\newcommand{\F}{\mathcal{F}}
\newcommand{\Tan}{\mathcal{T}}
\newcommand{\D}{\mathscr{D}}
\newcommand{\dif}{\mathrm{D}}
\renewcommand{\L}{\mathscr{L}}
\newcommand{\C}{\mathscr{C}}
\newcommand{\Simp}{S}
\newcommand{\simp}{\mathbf{s}}
\newcommand{\LSimp}{LS}
\newcommand{\lsimp}{\mathbf{ls}}
\newcommand{\dom}{\mathbf{dom}}
\newcommand{\codom}{\mathbf{cod}}

\newcommand{\dig}{\mathbf{p}}

\newcommand{\dere}{\mathbf{d}}

\newcommand{\weak}{\mathbf{w}}

\newcommand{\cont}{\mathbf{c}}

\newcommand{\tensor}{\otimes}

\begin{document}

% If the title is longer than 55 characters, then specify a shorter running title as the optional argument to \title. The running title should be roughyl at most 55 characters:
\title[A Fibrational Perspective on DiLL]{A Fibrational Perspective on Differential Linear Logic}
%\titlecomment{{\lsuper*}OPTIONAL comment concerning the title, \eg,
 %if a variant or an extended abstract of the paper has appeared elsewhere.}
%\thanks{thanks, optional.}	%optional

% affiliations are numbered automatically with a, b, c (see below)
% use the optional argument to indicate the affiliation(s) of each author
% omit the argument if there is only one author, or only one affiliation
\author[J.~Koleilat]{Jad Koleilat \lmcsorcid{0009-0009-1352-0455}}

% affiliation 1 (automatically numbered a)
\address{Université Sorbonne Paris Nord, LIPN, CNRS, UMR 7030, F-93430, Villetaneuse, France}	%optional
% write emails for all authors having that affiliation
\email{jad.koleilat@lipn.univ-paris13.fr}  %optional

%% etc.

%% required for running head on odd and even pages, use suitable
%% abbreviations in case of long titles and many authors:

%%%%%%%%%%%%%%%%%%%%%%%%%%%%%%%%%%%%%%%%%%%%%%%%%%%%%%%%%%%%%%%%%%%%%%%%%%%

%% the abstract has to PRECEDE the command \maketitle:
%% be sure not to issue the \maketitle command twice!

\begin{abstract}
 \noindent Differential Linear Logic (DiLL) is a sequent calculus that expresses differentiation via symmetries between linear and non-linear formulas. In this paper, we express categorical models of DiLL as a pair of Grothendieck fibrations equipped with a tangent functor. To do so, we adapt methods from categorical semantics of type theory to linear-non-linear adjunctions. This is a first step towards unifying DILL and dependent types.
\end{abstract}

\maketitle

%% start the paper here:
\section{Introduction}

Differential Linear Logic (DiLL) \cite{ehrhard_differential_2006}\cite{ehrhard_introduction_2018} is a sequent calculus capturing the notion of higher order differentiation from a logical perspective. It is an extension of Linear Logic (LL) \cite{girard_linear_1987} with three additional rules: cocontraction, coweakening, and codereliction. As their name suggests, those rules are symmetric to the contraction, weakening, and dereliction rules already present in LL. The differentiation of DiLL is related to the usual mathematical notion, as best illustrated by the convenient vector space model \cite{blute_convenient_2012} and the duality results between Cartesian Differential Categories (CDC)  
and categorical models of DiLL \cite{garner_cartesian_2021}\cite{blute_cartesian_2009}. \\
Differentiation, in all its generality, is a dependent construction: the differential of a smooth function $f : M \to N$ is a dependent function $D_x(f) : TM_x \multimap TN_{f(x)}$. Yet, DiLL is a simply typed theory; this raises the question, does a dependent version of DiLL exist? In this paper, we take a first step towards answering this question by expressing its models in a fibred way. \\ 
Given a smooth function $f : M \to N$, its differential $D(f)$ takes as inputs a point $x$ in $M$, a small variation $v_x$ at point $x$, and outputs a vector $D_x(f)(v_x)$ quantifying how much the variation in input $v_x$ influences the output $f(x)$. This is a dependent construction as the vector space of \enquote{small variations at point $x$}, denoted $TM_x$ and called the tangent space of $M$ at $x$, is a dependent family of vector spaces. This mix of dependence and linearity naturally guides us towards a dependent version of LL where one could express the type of $D(f)$ as $\Pi_{x : M} (TM_x \multimap TN_{f(x)})$.\\
Mixing dependency and linearity is a subtle subject, as linearity ensures variables only appear once, and type dependency allows variables to appear in terms and types \cite{krishnaswami_integrating_2015}\cite{pitts_categorical_2015}. Fortunately, we only need dependence on non-linear types to express the constructions of differential geometry. More explicitly, we need linear types depending on non-linear ones, but not the converse. Looking back at the typing of $D(f)$, there is a family of vector spaces (linear types), $TM$ and $TN$, depending on manifolds (non-linear types), $M$ and $N$. Such type theories exist, and we have an explicit description of their categorical semantics \cite{pitts_categorical_2015} \cite{lundfall_models_nodate} in terms of Grothendieck fibrations. \\
Fibrations are a categorical framework used to express dependency. A fibration is the data of two categories, $\mathbb{B}$ called the base category and $\mathbb{E}$ called the total category, as well as a functor $p : \mathbb{E} \to \mathbb{B}$ called the fibration. Intuitively, for every object $B \in \mathbb{B}$, $p^{-1}(B)$ is the category of objects depending on $B$. This formalism provides a way to interpret dependent typed theories \cite{jacobs_categorical_1999}. It can also provide a semantics for simply typed theories, as simply typed theories are particular instances of dependent typed theories where terms don't appear in any types. Concretely, models of simply typed theories can be described as a particular fibration called the simple fibration \cite{jacobs_categorical_1999}. 

\paragraph*{Contributions} 
From any model of LL expressed as a linear-non-linear adjunction \cite{mellies_categorical_2009}, we construct a fibration that we call the linear simple fibration. It is the linear logic equivalent of the simple category for simply typed theories. We show that any linear-non-linear adjunction corresponds to a fibred linear-non-linear adjunction between the simple fibration and the linear simple fibration. Moreover, requiring the existence of a particular section of this fibration, which we call a linear tangent functor, yields a model of DiLL. This extends the semantics of DiLL from Seely categories to any additive linear-non-linear adjunction with biproduct (see def. \ref{def:AddLNL}). This opens the search for new kinds of models arising from tangent categories. If they exist, they are good candidates for what should be models of dependent DiLL.

\paragraph*{Related Works} M. Kerjean, M. Rogers, and V. Maestracci \cite{kerjean_functorial_2025} showed that much of the structure of differentiation in DiLL is a consequence of the existence of a so-called differentation functor defined on the co-slice of the cartesian category. They extend the semantics of DiLL to any linear-non-linear adjunction with biproducts. Their approach is not a fibrational one, but it is not far away, as it relies on the co-slice category above the terminal object. This allows them to express the differential at point 0 of any morphism and deduce most axioms of differential Seely categories, but not the monoidal rule. Their failure to deduce this rule is due to the lack of axioms describing the behavior of partial differential (which expresses differentiation in context). The idea presented in section \ref{sec:GDSC} of constructing the deriving transform from the unit of the adjunction is taken from their work. \\
G. Cruttwell, J. Gallagher, JS. P. Lemay and D. Pronk \cite{cruttwell_monoidal_2022} discuss fibrations associated with a monoidal differential category with Seely isomorphisms. Their construction $\L_{\oc}[X]$ is the particular case of our linear simple category in the setting of Seely categories. \\
There has been work done to unify the major categorical framework for differentiation: (reversed) Cartesian Differential Categories, (reversed) Generalised Differential Categories, (reversed) Tangent Categories, around the notion of fibration \cite{capucci_fibrational_2024}. Differential Categories, which models DiLL \cite{blute_differential_2006} \cite{blute_differential_2020} \cite{hutchison_differential_2007},  are not present in their paper. We contribute to their work by adding differential Seely categories to this list.\\
M. Vákár \cite{pitts_categorical_2015} and M. Lundfall \cite{lundfall_models_nodate} describe the categorical semantic of dependent LL, which we crucially use in this work.

\paragraph*{Outline} In section \ref{prelim}, we introduce the notion of Grothendieck fibration and recall the semantics of LL and DiLL. In section \ref{LNL_fib}, we introduce the linear simple category of a linear-non-linear adjunction. This allows the construction of a fibred adjunction from any linear-non-linear adjunction. In section \ref{sec:GDSC}, using this formalism, we define generalised differential Seely categories and explain how they generalise models of DiLL. In section \ref{sec:CDC}, we link our work to Cartesian differential categories. In the last section \ref{conclu}, we conclude and discuss future works.

\paragraph*{Conventions}  We denote by \enquote{$;$} the diagramatic composition. Let $f : A \to B$ and $g : B \to C$, then $f;g : A \to C$ is their composite. We assume the monoidal categories we consider are strict monoidal. This is done to simplify calculations. We discuss the implications of this hypothesis in our results when it is necessary.

\section{Preliminary}
\label{prelim}

We recall notions of fibred category theory as well as the semantics of LL and DiLL.

\subsection{Grothendieck Fibrations}

Grothendieck fibrations, also called fibrations, are a categorical structure useful in type theory for expressing dependency. A fibration is the data of two categories $\mathbb{B}$ called the base category, $\mathbb{E}$ called the total category, and a functor $p : \mathbb{E} \to \mathbb{B}$ such that morphisms in $\mathbb{B}$ lift \enquote{nicely} in $\mathbb{E}$. The notion of Cartesian morphism captures this. 

\begin{defi}[Cartesian Morphism \cite{jacobs_categorical_1999}]
Let $p : \mathbb{E} \to \mathbb{B}$ be a functor. A morphism $f : X \to Y$ in $\mathbb{E}$ is said to be cartesian over $u : I \to J$ in $\mathbb{B}$ if $pf = u$ and for every $g : Z \to Y$ in $\mathbb{E}$ such that $pg = w;u$ for some $w : pZ \to I$, then there exists a unique $h : Z \to X$ in $\mathbb{E}$ above $w$, such that $h;f = g$. 
\begin{center}
\begin{tikzcd}[column sep=large, row sep=large]
Z \arrow[rd, "g"] \arrow[d, "h"', dotted] & \\
X \arrow[r, "f"'] & Y
\end{tikzcd}
\hspace{3em}
\begin{tikzcd}[column sep=large, row sep=large]
pZ \arrow[d, "w"'] \arrow[dr, "w;u = pg"] & \\
I \arrow[r, "u"'] & J
\end{tikzcd}
\end{center}
\end{defi}

\begin{defi}[Fibration \cite{jacobs_categorical_1999}]
A functor $p: \mathbb{E} \to \mathbb{B}$ is a fibration if for every $u : I \to J$ in $\mathbb{B}$ and for every choice of $Y \in \mathbb{E}$ above $J$, there exists a cartesian morphism $f : X \to Y$ above $u$. 
\end{defi}

\begin{defi}[Fibred Functor \cite{jacobs_categorical_1999}]
Let $p : \mathbb{E} \to \mathbb{B}$ and $q : \mathbb{D} \to \mathbb{B}$ be two fibrations, a fibred functor $F : p \to q$ is a functor $F : \mathbb{E} \to \mathbb{D}$ such that, $qF = p$ and $F$ preserves cartesian morphisms. 
\begin{center}
\begin{tikzcd}[row sep=large]
\mathbb{E} \arrow[dr, "p"'] \arrow[rr, "F"] & & \mathbb{D} \arrow[dl, "q"] \\
 & \mathbb{B}
\end{tikzcd}
\end{center}
\end{defi}

\begin{defi}[Fibre Above an Object \cite{jacobs_categorical_1999}]
For every fibration $p: \mathbb{E} \to \mathbb{B}$ and $J \in \mathbb{B}$, a morphism in $\mathbb{E}$ is said to be vertical above $J$ if it is sent to $\id_J$ by $p$. The fibre above $J$, denoted $\mathbb{E}_J$, is the category composed of the objects above $J$ and the vertical morphisms above $J$. 
\end{defi}

From a type theoretic perspective, one can think of the base category $\mathbb{B}$ as the category with types as objects and terms as morphisms. For example, a morphism $t : I \to J$ can be thought of as a term $t$ of type $J$ in context $I$. The fibre above an object $I$ is the category with objects in context depending on $I$. For example, the fiber above $\mathbf{Nat}$ might contain the context $n : \mathbf{Nat}, l : \mathbf{list}(n)$. A morphims in $\mathbb{E}_I$ between to context $i: I, x : \Gamma$ and $i: I, x' : \Gamma'$ is a term $i : I, x : \Gamma \vdash t : \Gamma'$.\\
Given a term $i: I \vdash t : J$, one obtains a functor from the fibre above $J$ to the fibre above $J$. This functor is denoted by $t^*$, takes as input a context depending on $J$; $j: J; x : \Gamma(j)$ and outputs the results of the substitution of the free variables of type $J$ by $t$; $i :I, x : \Gamma(t(i))$. This rewriting operation is captured categorically by the notion of reindexing functors.  

\begin{defi}[Cloven Fibration \cite{jacobs_categorical_1999}]
A fibration $p : \mathbb{E} \to \mathbb{B}$ is said to be cloven if it comes equipped with a choice of cartesian morphisms. Concretely, for every $u : I \to J$ in $\mathbb{B}$ and $Y \in \mathbb{E}$ above $J$, denote by $\overline{u}_Y: u^*(I) \to Y$ the chosen cartesian morphisms above $u : I \to pY$. A fibred functor is cloven if it preserves the choice of cartesian morphisms.
\end{defi}

\begin{prop}[Reindexing Functors \cite{jacobs_categorical_1999}]
\label{prop:ReindexingFunctors}
Let $p : \mathbb{E} \to \mathbb{B}$ be a cloven fibration then, every morphisms $u : I \to J$ in $\mathbb{B}$ induces a functor, called a reindexing functor, from $u^* :\mathbb{E}_J \to \mathbb{E}_I$. To build it, consider a morphism $f : X \to Y \in \mathbb{E}_J$. Then, $p (\overline{u}_X;f) = id_I; p(\overline{u}_Y)$. Since $\overline{u}_Y$ is a cartesian morphism above $u : I \to pY$, there exists a unique morphism denoted $u^*(f) : u^*(X) \to u^*(Y)$ above $\id_I$ such that $u^*(f); \overline{u}_Y = \overline{u}_X;f$. 
\begin{center}
\begin{tikzcd}[column sep=large, row sep=large]
u^*(X) \arrow[r, "\overline{u}_X"] \arrow[d, "u^*(f)"', dotted] & X \arrow[d, "f"] \\
u^*(Y) \arrow[r, "\overline{u}_Y"'] & Y
\end{tikzcd}
\end{center}
\end{prop}

\begin{rem}
\label{rem:FibFuncIndex}
Given two morphisms $u : I \to J$ and $v : J \to K$ in $\mathbb{B}$, one does not necessarily have $v^*; u^* = (u;v)^*$ and $\id_I^* = \id_{\mathbb{E}_I}$. We can show that there is a canonical isomorphisms $v^*; u^* \simeq (u;v)^*$ and $\id_I^* \simeq \id_{\mathbb{E}_I}$. Fibred functors preserve reindexing functors. 
\end{rem}

\begin{rem}
\label{rem:FibreFactor}
Every morphisms $f : X \to Y$ in $\mathbb{E}$ can be factored as a morphism $f' : X \to f^*(Y)$ in $\mathbb{E}_{pX}$ composed with $\overline{f}_Y$. 
\end{rem}

\begin{defi}[Split Fibration \cite{jacobs_categorical_1999}]
\label{def:SplitFib}
A fibration $p : \mathbb{E} \to \mathbb{B}$ is said to be split if it is cloven and for every $u : I \to J$ and $v : J \to K$ in $\mathbb{B}$, the canonical isomorphisms $v^*; u^* \simeq (u;v)^*$ and $\id_I^* \simeq \id_{\mathbb{E}_I}$ are equalities. In other words, in a split fibration $v^*; u^* = (u;v)^*$ and $\id_I^* = \id_{\mathbb{E}_I}$. A functor of fibrations is split if it preserves the splitting. 
\end{defi}

An important notion in the theory of fibrations is fibred structures. By this, we mean a structure on each fibre that is preserved by reindexing functors.

\begin{defi}[Fibred Cartesian Product \cite{jacobs_categorical_1999}]
A fibration $p : \mathbb{E} \to \mathbb{B}$ has fibred cartesian products if for each object $J$ in $\mathbb{B}$ the fibre $\mathbb{E}_J$is a cartesian category. Moreover, we require that every reindexing functor preserves the cartesian structure.
\end{defi}

The matter of fibred monoidal categories is more subtle, as one can require more or less structure to be preserved more or less strongly. Hence, there exist different definitions in particular monoidal fibrations \cite{shulman_framed_2008} and lax-monoidal fibrations \cite{zawadowsk_lax_2011}. Monoidal fibration is the stronger notion as it implies the second one. In this paper, the fibred monoidal categories we consider are monoidal fibrations. 

\begin{defi}[Monoidal Fibration \cite{shulman_framed_2008}]
\label{def:MonoidalFibration}
A fibration $p : \mathbb{E} \to \mathbb{B}$ is a monoidal fibration if:
\begin{itemize}
\item $\mathbb{E}$ and $\mathbb{B}$ are monoidal categories.
\item $p$ is a strict monoidal functor
\item the monoidal product of $\mathbb{E}$ preserves cartesians arrows. 
\end{itemize}
\end{defi}

\subsection{Linear-non-linear Adjunctions}

At the core of LL lies the distinction between linear proofs and non-linear ones. Both notions are linked by the introduction of the  \enquote{of course} $\oc$ connective, signalling when a resource can be duplicated or erased arbitrarily. The introduction of $\oc$ forces the disctinction between ditive conjunction $\with$ and a multiplicative one $\tensor$, and a linear implication $\multimap$ and a non-linear one $\implies$ (defined as $\oc(-) \multimap -$). The categorical framework capturing this structure is called a Linear-Non-Linear adjunction (LNL) \cite{mellies_categorical_2009}. A remarkable feature of this structure is that the monoidal closure (required to interpret $\multimap$), the additive conjunction (required to interpret $\with$), and the $*$-autonomous structure (required to interpret negation in classical LL) are are independant and sufficient features. By this, we mean that simply requiring their existence is enough for all the required compatibility conditions to be fulfilled. As such, a model of Intuitionistic Linear Logic (ILL) with additive conjunction is simply a model of ILL that has (categorical) products. 

\begin{defi}[Linear-non-linear adjunction]
\label{lnl}
A linear-non-linear adjunction is a symmetric monoidal adjunction between a symmetric monoidal category $(\L, \tensor, 1)$ and a cartesian category $(\C, \times, I)$ 
\begin{center}
\begin{tikzcd}[column sep=large]
(\C, \times, I)
\arrow[r, "{(\F, m)}"{name=F}, bend left=25] &
(\L, \tensor, 1)
\arrow[l, "{(\U, n)}"{name=G}, bend left=25]
\arrow[phantom, from = F, to = G, "\dashv" rotate=-90]
\end{tikzcd}
\end{center}
\end{defi}

\begin{rem}
This definition differs from the one traditionally found in the literature \cite{mellies_categorical_2009} in that we do not require $(\L, \tensor, 1)$ to be monoidal closed.  
\end{rem}

Intuitively, the monoidal category $\L$ is the \enquote{world of linear resources}, and the cartesian category $\C$ is the \enquote{world of non-linear resources}. In $\L$, objects cannot be freely duplicated or erased, whereas in $\C$, every object is a cocommutative comonoid (every object can be duplicated or erased). The monoidal adjunction enables resources to be considered as duplicable or non-duplicable. We expand on this idea later in this subsection. \\
Throughout this paper, $\C$ will be called the cartesian category and $\L$ the monoidal category of the LNL. Denote by $\eta : \id_\C \to \U \comp \F$ the unit of the adjunction and $\dere : \F \comp \U \to \id_\L$ the counit. Define $\oc := \F \comp \U$ the comonad on $\L$ induced by the adjunction with counit $\dere : \oc \to \id_\L$ and comultiplication $\dig : \oc \to \oc \oc$. 

\begin{prop}
\label{m_is_colax}
Given an LNL adjunction as above, $(\F, m)$ is a strong monoidal functor with the oplax structure constructed from $n$ as in proposition \ref{lax-colax_duality}.
\end{prop}

\begin{prop}[Lax-colax duality \cite{mellies_categorical_2009}]
\label{lax-colax_duality}
Given an adjunction
\begin{center}
\begin{tikzcd}[column sep=large]
(\C, \times, I)
\arrow[r, "{\F}"{name=F}, bend left=25] &
(\L, \tensor, 1)
\arrow[l, "{\U}"{name=G}, bend left=25]
\arrow[phantom, from = F, to = G, "\dashv" rotate=-90]
\end{tikzcd}
\end{center}
If $(\U, n)$ is a lax monoidal functor, then we can endow $\F$ with a colax structure $p$. This colax structure is given by :
\begin{equation*}
p_{X,Y} :=   \F(\eta_X \times \eta_Y); \F(n_{\F(X), \F(Y)}); \dere_{\F(X) \tensor  \F(Y)}  
\end{equation*}
\begin{equation*}
p_I := \F(n); \dere_1
\end{equation*}
Dualy, if $(\F, p)$ is colax monoidal functor then we can endow $\U$ with a lax monoidal structure $n$ defined as follows :
\begin{equation*}
n_{A,B} :=  \eta_{\U(A) \times \U(B)}; \U(p_{\U(A), \U(B)});  \U(\dere_A \tensor \dere_B)
\end{equation*}
\begin{equation*}
n_I := \eta_I; \U(p)
\end{equation*}
The first transformation defines a function $\phi$ from the lax monoidal structures on $\U$ to the colax monoidal structures on $\F$. The second transformation defines a function as well, which we can prove to be the inverse of $\phi$.
\end{prop}

The interpretation of the connectives of ILL is straightforward with the multiplicative conjunction $\tensor$ being interpreted by the monoidal product of $\L$ and the \enquote{of course} modality $\oc$ being interpreted by the induced comonad $\oc$ on $\L$.  Rules of ILL are interpreted as natural transformations in $\L$. For example, the dereliction rule is interpreted using the counit of the comonad $\dere$. The contraction and weakening rules are respectively interpreted using natural transformations with components $\cont_A : \oc A \to \oc A \tensor \oc A$ and $\weak_A : \oc A \to 1$. This makes any object $\oc A$ into a cocommutative comonoid. 

\begin{defi}[Contraction and Weakening]
\label{con_weak_def}
To construct them, use the colax structure on $\F$ (prop. \ref{m_is_colax}) to transport the cocommutative comonoids $(X, \Delta_X, I_X)$ from $\C$ to $\L$ where $\Delta_X : X \to X \times X$ is the diagonal morphism and $I_X : X \to I$ is the terminal morphism. Explicitly:
\begin{center}
\begin{tikzcd}[column sep=large]
\cont_{X} := \F(X) \arrow[r, "\F(\Delta_X)"] & \F(X \times X) \arrow[r, "m^{-1}"] & \F(X) \tensor \F(X)
\end{tikzcd}
\begin{tikzcd}[column sep=large]
\weak_{X} := \F(X) \arrow[r, "\F(I_X)"] & \F(I) \arrow[r, "m^{-1}"] & 1 
\end{tikzcd}
\end{center} 
For every $X \in \C$, $(\F(X), \cont_X, \weak_X)$ is a cocommutative comonoid. Write $\cont_A$ and $\weak_A$ respectively for $\cont_{\U(A)}$ and $\weak_{\U(A)}$ which are the announced contraction and weakening natural transformations. 
\end{defi}

The $\oc$ modality of Linear Logic allows one to duplicate or erase at will a resource. This manifests in the semantic by $\oc A$ being a cocommutative comonoid in $(\L, \tensor, 1)$. A morphism $f : A \to B$ is thought of as a linear morphism as it consumes one resource $A$ and outputs $B$. On the other hand, a morphism of type $f : \oc A \to B$ is thought of as a non-linear morphism since it can use the resource $A$ arbitrarily many times. Moreover, such morphisms are in bijection, through the adjunction, with morphisms in $\C$ (the non-linear world). 

\begin{defi}
\label{def:dagger}
Let $f : \F(X) \to A$ in $\L$, define $f^\dagger$ as the image of $f$ throughout the adjunction: $f^\dagger := \eta_{X}; \U(f) : X \to \U(A)$.
\end{defi}

In Linear Logic, the sequent $\oc A \tensor \oc B \vdash \oc(A \tensor B)$ is provable. Categorically, this corresponds to $\oc$ being a (lax) monoidal functor. The compatibilities between the lax structure and the comonad structure are satisfied, making $\oc$ into a lax monoidal comonad \cite{blute_differential_2020}.

\begin{prop}
Given an LNL, define $m_\tensor$ as:
\begin{center}
\begin{tikzcd}[column sep=huge, row sep=large]
{m_\tensor}_{A,B} := \oc A \tensor \oc B \arrow[r, "m_{\U(A), \U(B)}"] & \F(\U(A) \times \U(B)) \arrow[r, " \F(n_{A,B})"] & \oc (A \tensor B)
\end{tikzcd}
\begin{tikzcd}[column sep=huge, row sep=large]
{m_\tensor}_{1} := 1 \arrow[r, "m_{1}"] & \F(I) \arrow[r, " \F(n_{I})"] & \oc 1
\end{tikzcd}
\end{center}
$(\oc, m_\tensor)$ is a symmetric (lax) monoidal endofuctor.
\end{prop}

Seely categories are an important part of the story of this paper as they provide the basis for the semantic of DiLL with products. A Seely category is a particular instance of an LNL where the monoidal category has products. As a direct consequence, it is involved in an LNL with the Kleisli category of the comonad $\oc$. 

\subsection{Differential Seely Categories}

Differential Seely Categories\footnote{In the literature, they are called Monoidal Storage Differential Categories.} \cite{blute_differential_2006} \cite{blute_differential_2020} (DSC) are the categorical descriptions of the semantic of DiLL with products. They are enriched categories over commutative monoids (hence, the product is a biproduct). This condition is required to interpret the derivative of $\cont$, which corresponds to the product rule. 

\begin{defi}[Additive monoidal category \cite{blute_differential_2020}]
\label{add_mon}
An additive monoidal category is a symmetric monoidal category $(\L, \tensor, 1)$ enriched over $\mathrm{CMon}$ (the category of commutative monoids) such that this enrichment is compatible with the monoidal structure:
\begin{equation*}
k \tensor (f+g) \tensor h = k \tensor f \tensor h + k \tensor g \tensor  h \quad \quad k \tensor 0 \tensor h = 0
\end{equation*}
\end{defi}

\begin{defi}[Additive Linear-non-linear Adjunction]
\label{def:AddLNL}
A LNL is additive if its monoidal category $(\L, \tensor, 1)$ is an additive monoidal category.
\end{defi}

\begin{defi}[Differential Seely Categories \cite{blute_differential_2020}]
\label{def:DSC}
A differential Seely category is an additive Seely category equipped with a natural transformation $\partial_A : \oc A \tensor A \to \oc A$ called a deriving transform satisfying the following conditions: 
\begin{enumerate}[{[d.1]}]
\item Constant rule: $\partial_A; \weak_A = 0$
\item Leibniz rule: $\partial_A; \cont_A = \cont_A \tensor \id_A; (\id_{\oc A} \tensor \partial_A) + (\id_{\oc A} \tensor \sigma_{\oc A, A};\partial_A \tensor \id_{\oc A})$
\item Linear rule: $\partial_A; \dere_A = \weak_A \tensor \id_A$
\item Chain rule: $ \partial_A; \dig_A = \cont_A \tensor \id_A; \dig_A \tensor \partial_A; \partial_{\oc A}$
\item Interchange rule: $ \id_{\oc A} \tensor \sigma_{A,A}; \partial_A \tensor \id_A; \partial_A = \partial_A \tensor \id_A; \partial_A$
\end{enumerate}
\end{defi}

Given a non-linear morphism $f : \oc A \to B$, its differential is obtained by pre-composition with $\partial_A$. The differential of $f$ is $\partial_A; f: \oc A \tensor A \to B$. $\oc A \tensor A$ plays the role of the tangent bundle of $A$. Since $A$ is a linear resource, we think of it as a vector space. In differential geometry, the tangent bundle of a vector space is its product with itself, yet one side is considered as points in a manifold and the other as a (trivial) vector bundle. This is captured in $\L$ by the typing, $\oc A$ (non-linear $A$) represents the points of $A$ seen as a manifold, and $A$ is the vector space $A$. This is made more rigorous later in the article. \\
This list of axioms is the categorical expression of the usual mathematical properties of the differential. The constant rule expresses that the differential of a constant is zero, the Leibniz rule expresses the differential of a product of functions, the linear rule states that the differential of a linear function is the linear function itself, the chain rule expresses the differential of the composition of functions, and finally, the interchange rule expresses symmetries in the second differential. 

\section{Linear-non-linear Adjunctions as Fibrations}
\label{LNL_fib}

The purpose of this section is to show that the simple category of simply typed theories has an equivalent in LL that we call the linear simple category. We start by reminding of the simple category construction, then describe the linear simple category construction. We detail how the structure of the LNL lifts to the fibrations and finally how it relates to the semantics of dependent LL.

\begin{defi}[Simple category \cite{jacobs_categorical_1999}]
Let $(\mathbb{C}, \times, I)$ be a cartesian category, define the simple category $\Simp(\mathbb{C})$ as follows :
\begin{itemize}
\item Objects:  pairs $(X,J)$ of objects of $\mathbb{C}$.
\item Morphisms: $(f,u) : (X,J) \to (Y,K)$ are pairs of morphisms $f : X \to Y$ and $u : X \times J \to K$.
\item Composition is given using the comonoid structure induced by the product: 
\begin{equation*}
(f,u);(g,v) := (f;g \ , \  \Delta_X \times \id_J; f \times u; v)
\end{equation*}
More visually, the second component of $(f,u);(g,v)$ is
\begin{center}
\begin{tikzcd}[column sep=huge, row sep=large]
X \times J \arrow[r, " \Delta_X \times \id_J"] & X \times X \times J \arrow[r, "f \times u"] & Y \times K \arrow[r, "v"] & L
\end{tikzcd}
\end{center}
\item The identity $\id_{(X,J)}$ is $(\id_X, \pi_J)$. 
\end{itemize} 

The projection functor, called the simple fibration of $\mathbb{C}$, $\simp : \Simp(\mathbb{C}) \to \mathbb{C}$ sending $(X, J)$ to $X$ and $(f,u)$ to $f$ is a split fibration. The cartesian morphism above $f : X \to \simp(Y,K)$ is $(f, \pi_K) : (X,K) \to (Y,K)$. 
\end{defi}

Given a simply typed theory with products, one can interpret it in a cartesian category in the standard way: a term $\Gamma \vdash t : U$ is interpreted as a morphisms $\llbracket t \rrbracket : \llbracket \Gamma \rrbracket \to \llbracket U \rrbracket$ and a product of terms $\Gamma \vdash (u,v) : U \times V$ is interpreted as $\langle \llbracket u \rrbracket, \llbracket v \rrbracket \rangle :  \llbracket \Gamma \rrbracket \to \llbracket U \rrbracket \times  \llbracket V \rrbracket$. Keeping this in mind, an object $(X, J)$ in $\simp(\mathbb{C})$ is a context of shape $X, J$ and a morphism $(f,u) : (X,J) \to (Y,K)$ is a pair of terms $X \vdash f : Y$ and $X,J \vdash u : K$. Together, they form a term of type $X,J \vdash (f,u) : Y \times K$. This last step (crucial for composition) is not possible in LL as it relies on the full power of weakening and contraction. Before explaining our construction for LL, we discuss a bit more the structure of $\simp(\mathbb{C})$. 

\begin{rem}
\label{remark_abstract_non_sens}
There are two other ways to define this category. First, one can see it as the full subcategory of $\mathbb{C}^\rightarrow$ (the category of arrows) containing only the projections as objects. Secondly, let $F : \mathbb{C} \to \mathrm{coMonad}(\mathbb{C})$ be the functor sending an object $C$ of $\mathbb{C}$ to the comonad $C \times - : \mathbb{C} \to \mathbb{C}$ and $K : \mathrm{coMonad}(\mathbb{C})^{\op} \to \mathrm{Cat}$ be the contravarient functor sending a comonad $T$ of $\mathbb{C}$ to its Kleisli category $\mathbb{C}_T$. Then, the composite $K F^\op : \mathbb{C}^\op \to \mathrm{Cat}$ is an indexed category with the corresponding fibration being $\simp$. Moreover, since every Kleisli category $\mathbb{C}_{C \times -} $ can be equipped with a canonical cartesian product, $KF^\op$ is in fact a functor into the category of monoidal categories with strong monoidal functors. Thus, since $\mathbb{C}$ is a cartesian category, $\simp$ is a cartesian monoidal fibration \cite{shulman_framed_2008}.
\end{rem}

\begin{prop}
The fibrewise product of the simple fibration is given by:
\begin{itemize}
\item $(X, J) \times_X (X, K) = (X, J \times K)$.
\item $\pi_{(X,J)} := (\id_X, \pi_{J \times K} ;\pi_J) : (X, J \times K) \to (X, J)$.
\item $I^X := (X, I)$ is the unit.
\end{itemize} 
\end{prop}

\begin{rem}
$\Simp(\C)$ also has products given by $(X, J) \times (Y, K) = (X \times Y, J \times J)$. This makes $\simp$ into a strict cartesian monoidal functor and a cartesian monoidal fibration (definition \ref{def:MonoidalFibration}).
\end{rem}
 
As illustrated by the composition in $\simp(\mathbb{C})$, the left object in the pair $(X, J)$ must be a comonoid (a duplicable resource). The linear simple category construction is based on this observation, allowing only \enquote{non-linear} objects as the left elements of the pairs. 

\begin{defi}[Linear simple category]
\label{def_lnl_simp_cat}
Given an LNL 
\begin{center}
\begin{tikzcd}[column sep=large]
(\C, \times, I)
\arrow[r, "{\F}"{name=F}, bend left=25] &
(\L, \tensor, 1)
\arrow[l, "{\U}"{name=G}, bend left=25]
\arrow[phantom, from = F, to = G, "\dashv" rotate=-90]
\end{tikzcd}
\end{center}
define the linear simple category $\LSimp(\C)$ as follows:
\begin{itemize}
\item Objects: pairs $(X, A)$ with $X \in \C$ and $A \in \L$.
\item Morphisms:  $(f, u) : (X,A) \to (Y, B)$ are pairs of morphisms $f :  X \to Y$ in $\C$ and $u : \F(X) \tensor A \to B$ in $\L$.
\item Composition is given by: 
\begin{equation*}
(f,u); (g,v) := (f;g \ , \ \cont_X \tensor \id_A; \F(f) \tensor u; v)
\end{equation*}
More visually, the second component of $(f,u);(g,v)$ is
\begin{center}
\begin{tikzcd}[column sep=huge, row sep=large]
\F(X) \tensor A  \arrow[r, " \cont_X \times \id_A"] & \F(X) \tensor \F(X) \tensor A \arrow[r, "\F(f) \tensor u"] & \F(Y) \tensor B \arrow[r, "v"] & C
\end{tikzcd}
\end{center}
\item The unit $\id_{(X,A)}$ is $(\id_X, \weak_X \tensor \id_A)$. 
\end{itemize} 

The projection functor, called the linear simple fibration of the LNL,  $\lsimp : \LSimp(\C) \to \C$ sending $(X, A)$ to $X$ and $(f,u)$ to $f$ is a split fibration. The cartesian morphism above $f : X \to \lsimp(Y,B)$, denoted $\overline{f}_B$, is $(f, \weak_X \tensor \id_B) : (X,B) \to (Y,B)$. Denote by $f^*$ the reindexing functor induced by $f$.
\end{defi}

As per the simple category, an object $(X, A)$ in $\lsimp(\C)$ is a linear context $\F(X), A$. In the case where $X = \U(B)$ for some $B$ then $(\U(B), A)$ is simply the context $\oc B, A$. A morphism $(f, u) : (X,A) \to (Y, B)$ is a pair of terms $\F(X) \vdash f : \F(Y)$ and $\F(X), A \vdash u : B$. Since we can apply the contraction rule on $\F(X)$, we can form a term of type $\F(X), A \vdash t : \F(Y) \tensor B$. Another way to think of the category $\LSimp(\C)$ is as the category of partial linear maps in $\L$. A map in $\LSimp(\C)_X$ is a map $u$ of type $\F(X) \tensor A \to B$. Using the analogy that a non-linear map is a map of shape $\F(Y) \to C$, we see that $u$ is non-linear in $X$ and linear in $A$.

\begin{rem}
This category can be constructed in the same way as the simple category (see remark \ref{remark_abstract_non_sens}). One needs to replace the functor $F : \mathbb{C} \to \mathrm{coMonad}(\mathbb{C})$ by the functor $F' : \C \to \mathrm{coMonad}(\L)$ sending an object $X$ of $\C$ to the comonad $\F(X) \tensor - : \L \to \L$. Similarly, since $\L_{\F(X) \tensor -}$ can be endowed with a canonical monoidal product and since $\C$ is cartesian, $\lsimp$ is a symmetric monoidal fibration. 
\end{rem}

\begin{exa}
Let $\mathrm{Vect}$ be the monoidal category of real vector spaces with morphisms linear maps and monoidal structure given by the tensor product. There is an LNL between $\mathrm{Set}$ and $\mathrm{Vect}$
\begin{center}
\begin{tikzcd}[column sep=large]
(\mathrm{Set}, \times, 1)
\arrow[r, "{\F}"{name=F}, bend left=25] &
(\mathrm{Vect}, \tensor, \mathbb{R})
\arrow[l, "{\U}"{name=G}, bend left=25]
\arrow[phantom, from = F, to = G, "\dashv" rotate=-90]
\end{tikzcd}
\end{center}
with $\U$ the forgetfull functor and $\F$ the free vector space functor. In this case, the linear simple category is described as follows:
\begin{itemize}
\item Objects: pairs $(X,V)$ with $X$ a set and $V$ a vector space.
\item Morphisms: a morphism from $(X,V)$ to $(Y, W)$ is a pair $(f,u)$ with $f : X \to Y$ a function and $u : X \times V \to W$ a function such that for each $x \in X$ the function $u(x, -) : V \to W$ is a linear map.
\end{itemize}
\end{exa}

\begin{prop}
\label{tensor_LS}
The fibrewise symmetric monoidal product of the simple fibration is given by : 
\begin{itemize}
\item $(X, A) \tensor_X (X, B) = (X,  A \tensor B)$
\item $1^X := (X, 1)$ is the unit.
\item The monoidal product of two arrows $(\id_X ,u) : (X,A) \to (X, A')$ and $(\id_X ,v) : (X,B) \to (X, B')$ is given by 
\begin{equation*}
(\id_X, u) \tensor_X (\id_X, v) = (\id_X,  \cont_X \tensor \id_{A \tensor B}; \id_{\F(X)} \tensor \sigma_{\F(X), A} \tensor \id_B; u \tensor v)
\end{equation*}
More visually:
\begin{center}
\begin{tikzcd}[column sep=huge, row sep=large]
\F(X) \tensor A \tensor B  \arrow[r, " \cont_X \tensor \id_{A \tensor B}"] & \F(X) \tensor \F(X) \tensor A \tensor B  \arrow[d, "\id_{\F(X)} \tensor \sigma_{\F(X), A} \tensor \id_B"'] & \\
& \F(X) \tensor A \tensor \F(X) \tensor B \arrow[r, "u \tensor v"]  &  A' \tensor B'
\end{tikzcd}
\end{center}
\end{itemize} 
Moreover, each fiber is a strict symmetric monoidal category, and the reindexing functors are strict symmetric monoidal. Additionally, $\LSimp(\C)$ is a monoidal category with the monoidal product given by $(X,A) \tensor (Y,B) = (X \times Y, A \tensor B)$. This makes $\lsimp$ into a monoidal fibration (definition \ref{def:MonoidalFibration}).
\end{prop}

\begin{rem}
This construction is a strict generalisation of the simple category, as every cartesian category is in a trivial linear-non-linear adjunction with itself. In this case, the linear simple category is exactly the simple category. It can also be seen as a generalisation from Seely categories to LNL of the construction $\mathcal{L}_{\oc}[X]$ found in the literature of Differential Categories \cite{cruttwell_monoidal_2022}.
\end{rem}

\begin{prop}
\label{prop:ProductLS}
The simple linear fibration of an LNL with products $(\L, \with, \top)$ has fiberwise cartesian products:
\begin{itemize}
\item $(X, A) \with_X (X, B) = (X,  A \with B)$
\item $\top^X := (X, \top)$ is the unit.
\item The monoidal product of two arrows $(\id_X ,u) : (X,A) \to (X, A')$ and $(\id_X ,v) : (X,B) \to (X, B')$ is given by 
\begin{equation*}
(\id_X, u) \tensor (\id_X, v) = (\id_X, \mathbf{dist}_{\F(X),A,A'}; u \with v)
\end{equation*}
with $\mathbf{dist}_{A,B,C}$ the canonical morphism from $A \tensor (B \with C)$ to $(A \tensor B) \with (A \tensor C)$. The reindexing functors are strict cartesian functors.
\end{itemize} 
Additionally, $\LSimp(\C)$ is a cartesian category with cartesian structure given by $(X, A) \with (Y, B) = (X \times Y, A \with B)$. Denote by $\pi^X_i$, $i = 1,2$ the $i$-th projection in the fiber above $X$. 
\end{prop}

We have shown that from any LNL we can construct two fibrations on the cartesian category $\C$. We now have to describe how they interact.

\begin{center}
\begin{tikzcd}[column sep=large, row sep=large]
\Simp(\C) \arrow[dr, "\simp"']  &  & \LSimp(\C) \arrow[dl, "\lsimp"] \\
& \C  & 
\end{tikzcd}
\end{center} 

Linear-non-linear adjunctions describe a dialogue between a Cartesian category and a monoidal one. This lifts to a fibred adjunction between the simple category of $\C$ and the linear simple fibration. To construct this adjunction, we must first show that the functors $\U$ and $\F$ induce fibred functors between $\simp$ and $\lsimp$. The idea is to apply $\U$ and $\F$ \enquote{pointwise}, this also ensures that the lax structures on $\U$ and $\F$ are preserved fibrewise. 

\begin{prop}
\label{U_simp_func}
Define the functor $\U^\Simp : \LSimp(\C) \to \Simp(\C)$ in the following way:
\begin{itemize}
\item On objects: $\U^\Simp(X,A) := (X, \U(A))$.
\item On morphisms: let $(f,u) : (X, A) \to (Y, B)$, $\U^\Simp(f,u) := (f, \eta_X \times \id_{\U(A)}; n_{\F(X),A}; \U(u))$. More visualy, the second component of $\U^\Simp(f,u)$ is
\begin{center}
\begin{tikzcd}[column sep=huge, row sep=large]
X \times \U(A)  \arrow[r, " \eta_X \times \id_{\U(A)}"] & \U(\F(X)) \times \U(A) \arrow[r, "n_{\F(X),A}"] & \U(\F(X) \tensor A) \arrow[r, "\U(u)"] & \U(B)
\end{tikzcd}
\end{center}
\end{itemize} 
$\U^\Simp$ is a split fibered functor from $\lsimp$ to $\simp$. Additionally, it's fibrewise a symmetric lax-monoidal functor with the monoidal structure above $X$ given by:
\begin{itemize}
\item $n_{I^X} := \pi_I; n_I : I^X \to \U^\Simp(1^X)$.
\item  $n_{(X,A),(X,B)} := \pi_{\U(A) \times \U(B)}; n_{A,B} : \U^\Simp((X,A)) \times_X \U^\Simp((X,B)) \to \U^\Simp((X, A) \tensor_X (X,B))$. 
\end{itemize}
\end{prop}

The tedious part is proving that $\U^\Simp$ is indeed a functor. Once this is done, a direct check allows one to see that it is a lax-monoidal functor. The proof is split into a few lemmas to make it more readable.
 
\begin{lem}
\label{cont_dagger}
For all $X \in \C$, $\eta_X;\U(\cont_X) = \eta_X; \Delta_{\U(\F(X))}; n_{\F(X), \F(X)}$. 
\end{lem}

\begin{proof}
\begin{align*}
\eta_X; \U(\cont_X) &\underset{\ref{con_weak_def}}{=} \eta_X; \U(\F(\Delta_X); m^{-1}_{X,X}) \\
&\underset{\ref{m_is_colax}}{=} \eta_X; \U(\F(\Delta_X; \eta_X \times \eta_X; n_{\F(X),\F(X)}); \dere_{\F(X) \tensor \F(X)}) \\
&\underset{\hspace{1.4em}}{=} \eta_X; \U(\F(\eta_X; \Delta_{\U(\F(X))}; n_{\F(X),\F(X)}); \dere_{\F(X) \tensor \F(X)})  \\
&\underset{\hspace{1.4em}}{=} \eta_X; \Delta_{\U(\F(X))}; n_{\F(X),\F(X)}; \eta_{\U(\F(X) \tensor \F(X))}; \U( \dere_{\F(X) \tensor \F(X)}) \\
&\underset{\hspace{1.4em}}{=} \eta_X; \Delta_{\U(\F(X))}; n_{\F(X),\F(X)} \qedhere
\end{align*}
\end{proof}

\begin{lem}
\label{U_weak}
Let $f : X \to Y$ in $\C$ and $u : A \to B$ in $\L$ then, $\U^\Simp(f, \weak_X \tensor u) = (f, \pi_2; \U(u))$. 
\end{lem}

\begin{proof}
\begin{align*}
\U^\Simp((f, \weak_X \tensor u)) &\underset{\hspace{1.4em}}{=} (f, \eta_X \times \id_{\U(A)}; n_{\F(X),A}; \U(\weak_X \tensor u ) ) \\
&\underset{\hspace{1.4em}}{=} (f, \eta_X \times \id_{\U(A)}; \U(\weak_X) \times \U(u); n_{1, B}) \\
&\underset{\ref{con_weak_def}}{=} (f, \eta_X \times \id_{\U(A)}; \U(\F(I_X); m^{-1}_I) \times  \U(fu; n_{1, B}) \\
&\underset{\hspace{1.4em}}{=} (f, [\eta_X; \U(\F(I_X))] \times  \id_{\U(A)};  \U(m^{-1}_{I}) \times \U(u); n_{1, B} ) \\
&\underset{\hspace{1.4em}}{=} (f, (I_X; \eta_I) \times  \id_{\U(A)};  \U(m^{-1}_I) \times  \U(u); n_{1, B}) \\
&\underset{\hspace{1.4em}}{=} (f, I_X \times \id_{\U(A)} ; (\eta_I;\U(m^{-1}_I)) \times \U(u); n_{1, B}) \\
&\underset{\ref{lax-colax_duality}}{=} (f, I_X \times \id_{\U(A)} ; n_I \times \U(u); n_{1, B})\\
&\underset{\hspace{1.4em}}{=}(f, I_X \times  \U(u) ; \pi_2)\\
&\underset{\hspace{1.4em}}{=} (f,  \pi_2; \U(u)) \qedhere
\end{align*}
\end{proof}

\begin{proof}[Proof of Prop. \ref{U_simp_func}]
Applying Lemma \ref{U_weak} immediately yields that the identity is sent to the identity. We check that $\U^\Simp$ sends the composition to the composition:  let $(f,u) : (X,A) \to (Y,B)$ and $(g,v) : (Y,B) \to (Z,C)$,
\begin{align*}
&\U^\Simp(f,u);\U^\Simp(g,v) \\
&= (f  ,  \eta_X \times \id_{\U(A)}; n_{\F(X),A}; \U(u)); (g  , \eta_Y \times \id_{\U(B)}; n_{\F(Y),B}; \U(v))\\
&= (f;g , \Delta_X \times \id_A;  f \times  [\eta_X \times \id_{\U(A)}; n_{\F(X),A}; \U(u)]; \eta_Y \times \id_{\U(B)}; n_{\F(Y),B}; \U(v) )
\end{align*}
We continue the proof omitting the left component of the pair for more readability:
\begin{align*}
\Delta_X \times \id_A;  &f \times  [\eta_X \times \id_{\U(A)}; n_{\F(X),A}; \U(u)]; \eta_Y \times \id_{\U(B)}; n_{\F(Y),B}; \U(v) \\
&\underset{\hspace{1.4em}}{=} \Delta_X \times \id_{\U(A)}; (f; \eta_Y) \times [\eta_X \times \id_{\U(A)}; n_{\F(X),A}; \U(u)]; n_{\F(Y),B}; \U(v) \\
&\underset{\hspace{1.4em}}{=} \Delta_X \times \id_{\U(A)}; (\eta_X; \U(\F(f))) \times [\eta_X \times \id_{\U(A)}; n_{\F(X),A}; \U(u)]; n_{\F(Y),B}; \U(v) \\
&\underset{\hspace{1.4em}}{=} \Delta_X \times \id_{\U(A)}; \eta_X \times \eta_X \times \id_{\U(A)}; \U(\F(f)) \times (n_{\F(X),A}; \U(u)); n_{\F(Y),B}; \U(v) \\
&\underset{\hspace{1.4em}}{=} (\eta_X; \Delta_{\U(\F(X))}) \times \id_{\U(A)}; \U(\F(f)) \times (n_{\F(X),A}; \U(u)); n_{\F(Y),B}; \U(v) \\
&\underset{\hspace{1.4em}}{=} (\eta_X; \Delta_{\U(\F(X))}) \times \id_{\U(A)}; \id_{\U(\F(X))} \times n_{\F(X),A};  \U(\F(f)) \times \U(u); n_{\F(Y),B}; \U(v) \\
&\underset{\hspace{1.4em}}{=} (\eta_X; \Delta_{\U(\F(X))}) \times \id_{\U(A)}; \id_{\U(\F(X))} \times n_{\F(X),A}; n_{\F(X), \F(X) \tensor A}; \U(\F(f) \tensor u); \U(v) \\
&\underset{\hspace{1.4em}}{=}  (\eta_X; \Delta_{\U(\F(X))}) \times \id_{\U(A)}; n_{\F(X), \F(X)} \times \id_{\U(A)}; n_{\F(X) \tensor \F(X), A}; \U(\F(f) \tensor u; v) \\
&\underset{\hspace{1.4em}}{=} (\eta_X; \Delta_{\U(\F(X))}; n_{\F(X), \F(X)}) \times \id_{\U(A)}; n_{\F(X) \tensor \F(X), A}; \U(\F(f) \tensor u; v) \\
&\underset{\ref{cont_dagger}}{=} (\eta_X; \U(\cont_X)) \times \id_{\U(A)}; n_{\F(X) \tensor \F(X), A}; \U(\F(f) \tensor u; v)  \\
&\underset{\hspace{1.4em}}{=} \eta_X \times \id_{\U(A)};  \U(\cont_X) \times \id_{\U(A)} ; n_{\F(X) \tensor \F(X), A}; \U(\F(f) \tensor u; v) \\
&\underset{\hspace{1.4em}}{=} \eta_X \times \id_{\U(A)};  n_{\F(X), A}; \U(\cont_X \tensor \id_A);  \U(\F(f) \tensor u; v) \\
&\underset{\hspace{1.4em}}{=} \eta_X \times \id_{\U(A)};  n_{\F(X), A}; \U(\cont_X \tensor \id_A; \F(f) \tensor u; v) \\
&\underset{\hspace{1.4em}}{=} \U^\Simp((f,u);(g,v))
\end{align*}
This shows that $\U^\Simp$ is a functor. As a direct consequence of Lemma \ref{U_weak}, it is a split fibered functor.
\end{proof}

\begin{prop}
\label{F_simp_func}
Define the functor $\F^\Simp : \Simp(\C) \to \LSimp(\C)$ in the following way:
\begin{itemize}
\item On objects: $\F^\Simp(X,J) := (X, \F(J))$.
\item On morphisms: let $(f,u) : (X, J) \to (Y, K)$, $\F^\Simp(f,u) := (f, m_{X,J}; \F(u))$. More visualy, the second component of $\U^\Simp(f,u)$ is
\begin{center}
\begin{tikzcd}[column sep=huge, row sep=large]
\F(X) \tensor \F(J)  \arrow[r, "{m_{X,J}}"] & \F(X \times J) \arrow[r, "\F(u)"] & \F(K)
\end{tikzcd}
\end{center}
\end{itemize} 
$\F^\Simp$ is a split fibered functor from $\simp$ to $\lsimp$. Additionally, it is fibrewise a symmetric lax-monoidal functor with the monoidal structure above $X$ given by : 
\begin{itemize}
\item $m_{1^X} := ; \weak_X \tensor \id_1; m_1: 1^X \to \F^\Simp(I^X)$.
\item  $m_{(X,J),(X,K)} := \weak_X \tensor \id_{\F(J) \tensor \F(K)}; m_{J,K} : \F^\Simp((X,J)) \tensor_X \F^\Simp((X,K)) \to \F^\Simp((X, J) \times_X (X,K))$. 
\end{itemize}
\end{prop}

As for the proof of functoriality of $\U^\Simp$, we first start by stating a lemma.

\begin{lem}
\label{F_pi}
For all $X,J \in \C$ and $f : J \to K$, $m_{X,J}; \F(\pi_J; f) = \weak_X \tensor F(f)$.
\end{lem}

\begin{proof}
\begin{align*}
m_{X,J}; \F(\pi_J) &= m_{X,J}; \F(I_X \times \id_J); \F(\pi_J; f) \\
&\underset{\hspace{1.4em}}{=}  \F(I_X) \tensor \F(\id_J); m_{I,J}; \F(\pi_J; f)  \\
&\underset{\hspace{1.4em}}{=}\F(I_X) \tensor \id_{\F(J)}; m_I^{-1} \tensor \id_{\F(J)}; m_I \tensor \id_{\F(J)} ;  m_{I,J}; \F(\pi_J; f) \\
&\underset{\hspace{1.4em}}{=} \F(I_X) \tensor \id_{\F(J)}; m_I^{-1} \tensor \F(f) \\
&\underset{\hspace{1.4em}}{=}  \weak_X \tensor \F(f) \qedhere
\end{align*}
\end{proof}

\begin{proof}[Proof of Prop. \ref{F_simp_func}]
Applying Lemma \ref{F_pi} immediately yields that the identity is sent to the identity. We check that $\F^\Simp$ sends the composition to the composition:  let $(f,u) : (X,J) \to (Y,K)$ and $(g,v) : (Y,K) \to (Z,L)$,
\begin{flalign*}
\F^\Simp(f,u);\F^\Simp(g,v)&= (f  ,  m_{X,J};\F(u)); (g  ,  m_{Y,K};\F(v))\\
&= (f;g , \cont_X \tensor \id_{\F(J)}; \F(f) \tensor [m_{X,J};\F(u)] ; m_{Y,K};\F(v))
\end{flalign*}
We continue the proof, omitting the left component of the pair for more readability:
\begin{flalign*}
\cont_X \tensor \id_{\F(J)}; &\F(f) \tensor [m_{X,J};\F(u)] ; m_{Y,K};\F(v) \\
&\underset{\hspace{1.4em}}{=}  \cont_X \tensor m_{X,J}; \F(f) \tensor \F(u);  m_{Y,K};\F(v) \\
&\underset{\hspace{1.4em}}{=}  \cont_X \tensor m_{X,J}; m_{X, X \times J}; \F(f \times u); \F(v)  \\
&\underset{\hspace{1.4em}}{=}  \cont_X \tensor \id_{\F(J)}; m_{X,X} \tensor \id_{\F(J)}; m_{X \times X, J};  \F(f \times u); \F(v)  \\
&\underset{\ref{con_weak_def}}{=}  [\F(\Delta_X); m_{X,X}^{-1}] \tensor \id_{\F(J)}; m_{X,X} \tensor \id_{\F(J)}; m_{X \times X, J};  \F(f \times u); \F(v) \\
&\underset{\hspace{1.4em}}{=} \F(\Delta_X )\tensor \id_{\F(J)}; m_{X \times X, J};  \F(f \times u); \F(v) \\
&\underset{\hspace{1.4em}}{=} m_{X,J}; \F(\Delta_X \times \id_J); \F(f \times u); \F(v)  \\
&\underset{\hspace{1.4em}}{=} m_{X,J}; \F(\Delta_X \times \id_J; f \times u; v) \\
&\underset{\hspace{1.4em}}{=} \F^\Simp((f,u);(g,v))
\end{flalign*}
This shows that $\F^\Simp$ is a functor. As a direct consequence of Lemma \ref{F_pi}, it is a split fibered functor.
\end{proof}

Now that the functors $\U$ and $\F$ have been lifted to the fibrations, we show that these lifts are involved in a fibred linear-non-linear adjunction.

\begin{defi}[Fibred Adjoints \cite{jacobs_categorical_1999}]
Let $p : \mathbb{E} \to \mathbb{B}$, $q : \mathbb{D} \to \mathbb{B}$ be two fibrations and $F : p \to q$, $U : q \to p$ be two fibred functors. We say that $F$ and $U$ are fibred adjoint functors when $\mathbb{E} \to \mathbb{D} : F \dashv U : \mathbb{D} \to \mathbb{E}$ is a 1-categorical adjunction such that the unit and counit are sent by the fibrations to the identity. 
\end{defi}

\begin{prop}
\label{U_F_lnl}
$\U^\Simp$ and $\F^\Simp$ are fibered adjoint functors $ \F^\Simp \dashv \U^\Simp$ and this adjunction is fibrewise a linear-non-linear adjunction. The unit and counit are respectively $\eta_{(X,J)} := (\id_X, \pi_J; \eta_J)$ and $\dere_{(X,A)} := (\id_X, \weak_X \tensor \dere_A)$. 
\end{prop}

\begin{center}
\begin{tikzcd}[column sep=3em, row sep=3em]
\Simp(\C) \arrow[dr, "\simp"'] \arrow[rr, "\F^\Simp"{name=U}, bend left=12, , shift left=0.5]  &  & \LSimp(\C)  \arrow[ll, "\U^\Simp"{name=D}, bend left=12, shift left=0.5] \arrow[phantom, from=U, to=D, "\dashv" rotate=-90]  \arrow[dl, "\lsimp"] \\
& \C & 
\end{tikzcd}
\end{center}

\begin{rem} 
\label{rem:LSIisoL}
Notice that the categories $\Simp(\C)_I$ and $\LSimp(\C)_I$ are respectively isomorphic to $\C$ and $\L$. Moreover,  $\U^\Simp_I$ and $\F^\Simp_I$ are respectively isomorphic to $\U$ and $\F$. This means that we can recover the original data of the LNL from this fibred construction; no information was erased. 
\end{rem}

From a syntactic point of view, it is expected that all constructions already present in the LNL lift to the fibres. If one can construct a term $A \vdash t : B$, then one can always add in the context a non-linear type $\F(X), A \vdash t : B$, lifting $t$ to a morphism above $X$.  As such, the proofs we did not detail are simply a matter of verifying that in each fibre everything behaves as in the LNL. Since this pair of functors is fibrewise an LNL, we can define in each slice a weakening and contraction natural transformation.

\begin{defi}
Define $\oc^\Simp := \F^\Simp \U^\Simp$, it's a fibred comonad on $\lsimp$. It sends objects $(X,A)$ to $(X, \oc A)$ and morphisms $(f, u) : (X,A) \to (Y,B)$ to $(f, \F(\eta_X) \tensor \id_{\oc A}; {m_\tensor}_{\F(X), A}; \oc u)$. The counit $\dere_{(X,A)} : (X, \oc A) \to (X, A)$ is $(\id_X, \weak_X \tensor \dere_A)$ and the comultiplication $\dig_{(X,A)} : (X, \oc  A) \to (X, \oc \oc  A)$ is $(\id_X, \weak_X \tensor \dig_A)$.
\end{defi}

\begin{center}
\begin{tikzcd}[column sep=3em, row sep=3em]
\Simp(\C) \arrow[dr, "\simp"'] \arrow[rr, "\F^\Simp"{name=U}, bend left=12, , shift left=0.5]  &  & \LSimp(\C) \arrow[loop right, "\oc"] \arrow[ll, "\U^\Simp"{name=D}, bend left=12, shift left=0.5] \arrow[phantom, from=U, to=D, "\dashv" rotate=-90]  \arrow[dl, "\lsimp"] \\
& \C & 
\end{tikzcd}
\end{center}

\begin{cor}[Fibred contraction and weakening]
\label{def_f_c_w}
Since $\F^\Simp$ is fibrewise a colax-monoidal functor (corollary \ref{m_is_colax}), it induces on each slice a contraction and weakening natural transformation. Let $X, J \in \C$, 
\begin{equation*}
\weak_{(X,J)} = (\id_X, \weak_X \tensor \weak_J) : \F^\Simp(X,J) \to 1^X
\end{equation*} 
\begin{equation*}
\cont_{(X,J)} = (\id_X, \weak_X \tensor \cont_J) : \F^\Simp(X,J) \to \F^\Simp(X,J) \tensor_X \F^\Simp(X,J)
\end{equation*}
For every $A \in \L$, denote $\weak_{(X,A)}$ for $\weak_{\U^\Simp(X,A)}$ and $\cont_{(X,A)}$ for $\cont_{\U^\Simp(X,A)}$. More explicitely, for each object $(X, \oc A)$ in $\LSimp(\C)$,
\begin{equation*}
\weak_{(X,A)} = (\id_X, \weak_X \tensor \weak_A) : (X, \oc A) \to (X, 1)
\end{equation*} 
\begin{equation*}
\cont_{(X,A)} = (\id_X, \weak_X \tensor \cont_A) : (X, \oc A) \to (X, \oc A \tensor \oc A)
\end{equation*}
\end{cor}

Now that we have finished with the basic constructions, we relate them to comprehension categories, which are models of DTT. Those constructions come from the fact that the simple linear category is a comprehension category  \cite{jacobs_categorical_1999} with coproducts (dependent sum). Moreover, this fibration is a model of dependent linear logic in the sense of M. Vákár \cite{pitts_categorical_2015} and M. Lundfall \cite{lundfall_models_nodate}. Knowledge of comprehension categories is not required, as we give explicit definitions.

\begin{defi}[Comprehension Functors \cite{jacobs_categorical_1999}]
\label{def:Comprehension}
Denote $\set{-}$ the functor from $\Simp(\C)$ to $\C$ sending $(X,J)$ to $X \times J$ and $(f,u)$ to $\langle \pi_X;f , u \rangle : X \times J \to Y \times K$. Denote by $\banana{-}$ the functor $\U^\Simp; \set{-}$ from $\LSimp(\C)$ to $\C$. 
\end{defi}

From a geometric perspective, $\set{-} := \dom \comp P : \mathbb{E} \to \mathbb{B}$ is the total space functor, sending a bundle above $X$ to its total space. From a logical perspective, consider a set $X$ and a predicate on this set $Q(x)$. $Q$ lives in the fibre above $X$ and the comprehension functor applied to $Q$, $\set{Q}$ is the subset of $X$ satisfying $Q$ (hence the name comprehension). In dependent type theory, $\set{-}$ is the categorical equivalent of the operation allowing us to go from a derivation $x :\gamma \vdash \tau(x) : \Type$ to a dependent context $x : \gamma, y : \tau(x)$. 

\begin{center}
\begin{tikzcd}[column sep=3em, row sep=3em]
\Simp(\C) \arrow[dr, "\set{-}"', shift right=1]  \arrow[dr, "\simp", shift left=1] \arrow[rr, "\F^\Simp"{name=U}, bend left=10, , shift left=0.5]  &  & \LSimp(\C)  \arrow[ll, "\U^\Simp"{name=D}, bend left=10, shift left=0.5] \arrow[phantom, from=U, to=D, "\dashv" rotate=-90]   \arrow[dl, "\banana{-}", shift left=1] \arrow[dl, "\lsimp"', shift right=1] \\
& \C & 
\end{tikzcd}
\end{center}

\begin{defi}
\label{linear_sigma_types}
For every $X,J \in \C$, let $\Sigma_{(X,J)} :  \LSimp(\C)_{X \times J} \to \LSimp(\C)_X$ be the functor given by:
\begin{itemize}
\item On objects: $\Sigma_{(X,J)} (X \times J, A) := (X, \F(J) \tensor A)$.
\item On morphisms: $(\id_{X \times J}, u) : (X \times J, A) \to (X \times J, B)$ to 
\begin{equation*}
\Sigma_{(X,J)}(\id_{X \times J}, u) := (\id_X, \id_{\F(X)} \tensor \cont_J \tensor \id_B; \sigma_{\F(X), \F(J)} \tensor \id_{\F(J) \tensor A}; \id_{\F(J)} \tensor ( m_{X, J} \tensor \id_A;  u) )\\
\end{equation*}
Denote $\Sigma_{(X,A)}$ for $\Sigma_{(X, \U(A))}$. More visualy the second component of $\Sigma_{(X,J)}(\id_{X \times J}, u) $ is 
\begin{center}
\begin{tikzcd}[column sep=5em, row sep=3em]
\F(X) \tensor \F(J) \tensor A \arrow[r, "\id_{\F(X)} \tensor \cont_J \tensor \id_B"] & \F(X) \tensor \F(J) \tensor \F(J) \tensor A \arrow[d, "\sigma_{\F(X), \F(J)} \tensor \id_{\F(J) \tensor A}"'] & \\
& \F(J) \tensor \F(X) \tensor \F(J) \tensor A \arrow[d, "\id_{\F(J)} \tensor m_{X, J} \tensor \id_A"'] & \\
&  \F(J) \tensor \F(X \times J) \tensor A \arrow[r, "\id_{\F(J)} \tensor u"]& \F(J) \tensor B 
\end{tikzcd}
\end{center}
\end{itemize} 
\end{defi}

This functor is the linear equivalent of the usual dependent sigma type. Since we are in a simply typed setting, this functor is simply the non-dependent monoidal product $\F^\Simp(X,J) \tensor_X -$. 

\begin{prop}
For every $X,J \in \C$, $\Sigma_{(X,J)}$ is left adjoint to $\pi_{1}^* : \LSimp(\C)_X \to \LSimp(\C)_{X \times J}$. The unit $\nu^{(X,J)}$ and counit $\mu^{(X,J)}$ of the adjunction $\Sigma_{(X,J)} \dashv \pi_1^*$ are given by:
\begin{itemize}
\item The unit is defined component-wise as: 
\begin{equation*}
\nu^{(X,J)}_{(X \times J, A)} := (\id_{X \times J}, m^{-1}_{X, J} \tensor \id_A; \weak_X \tensor \id_{\F(J) \tensor A}) : (X \times J, A) \to (X \times J, \F(J) \tensor A)
\end{equation*}
\item The counit is defined component-wise as:
\begin{equation*}
\mu^{(X,J)}_{(X, A)} := (\id_X, \weak_X \tensor \weak_J \tensor \id_A) : (X, \F(J) \tensor A) \to (X, A)
\end{equation*}
\end{itemize} 
Denote by $\nu^{(X,A)}$ and $\mu^{(X,A)}$ for $\nu^{\U^\Simp(X,A)}$ and $\mu^{\U^\Simp(X,A)}$. 
\end{prop}

\begin{rem}
\label{rem:SigmaContWeak}
As explained by M. Vákár \cite{pitts_categorical_2015}, the existence of linear sigma types (fibred coproducts) determines the comonad $\oc$ completely. For every $A \in \L$ we can define $\oc^\Simp (X,A) = (X, \oc A)$ as $\Sigma_{(X,A)} \pi_1^*(1^X) = (X, \oc A \tensor 1)$. Moreover, we can recover the contraction and the weakening respectively as the comultiplaction and counit of the comonad $\Sigma \pi_1^*$.  Concretely, in the fiber above $X \in \C$ the comultiplication 
\begin{equation*}
\delta^{(X,J)}_{(X,A)} : (X, \F(J) \tensor A) \to (X, \F(J) \tensor \F(J) \tensor A)
\end{equation*} 
is equal to $\cont^X_J \tensor_X \id_{(X,A)}$ and the counit 
\begin{equation*}
\mu^{(X,J)}_{(X, A)} : (X, \F(J) \tensor A) \to (X, A)
\end{equation*}
is equal to $\weak^X_J \tensor_X \id_{(X,A)}$. 
\end{rem}

$\Sigma$ allows us to internalise in a fiber the composition in $\LSimp(\C)$. To avoid boring the reader with unnecessary development, we give one last identity relating $\F$ and $\Sigma$ that will prove useful later. 

\begin{lem}
\label{lem:SigmaF}
For every $u : (X,J) \to (X, K)$ in $\Simp(\C)_X$ and every $v : (X \times K, A) \to (X \times K ,B)$ in $\LSimp(\C)_{X \times K}$, 
\begin{equation*}
\Sigma_{(X,J)}( \set{u}^*v); \mu^{(X,J)}_{(X, A)}  = \F^\Simp(u) \tensor_X \id_A; \Sigma_{(X, K)} v; \mu^{(X,K)}_{(X, A)} : (X, \F(J) \tensor A) \to (X, B)
\end{equation*}
\end{lem}

This comprehension structure allows us to push to the following intuition: $\LSimp(\C)$ is a structure on $\C$ describing a notion of partial linear map. We say that a map $X \times \U(A) \to \U(B)$ is linear in its second argument if it is of shape $\banana{u}; \pi_2$ for some $u \in \LSimp(\C)_X$. 

\begin{defi}[Linear Map]
\label{def:LinearMap}
A morphism $f : \U(A) \to \U(B)$ is linear if it is equal to $\U(u) : \U(A) \to \U(B)$ for some $l : A \to B$ in $\L$.
\end{defi}

\begin{defi}[Partial Linear Map]
\label{def:PartialLinear}
A morphism $f : X \times \U(A) \to \U(B)$ is linear in its second component if there exists a morphism $u: (X,A) \to (X,B)$ in $\LSimp(\C)_X$ such that $f = \banana{u}; \pi_2$. A morphism $f : \U(A) \times X \to \U(B)$ is linear in its first component if $\sigma_{X, \U(A)};f : X \times \U(A) \to \U(B)$ is linear in its second component. A morphism $f : \U(A) \times \U(B) \to \U(C)$ is bilinear if it is linear in its first and second components.
\end{defi}

\section{Generalised Differential Seely Category}
\label{sec:GDSC}

Looking at the categorical semantics of type theories \cite{jacobs_categorical_1999}, to go from a simply typed theory to a dependently typed theory, one needs to ask for structures in the fibers instead of in the base category. Now that we have seen that simply typed models of LL can be expressed as some kind of fibrations, we are going to show that requiring the existence of a particular functor will make each fiber into a model of DiLL. To do so, we take inspiration from differential geometry. \\

First, there is the intuition that LL should be a \enquote{logic of vector bundle} \cite{mellies_functorial_2022}. In our setting, the linear simple category is very similar in spirit to a category of (trivial) vector bundles. In differential geometry, differentiation corresponds to a functor (tangent functor) from the category of manifolds (non-linear types) to the category of vector bundles (dependent linear types). We follow this intuition and define a functor, called a linear tangent functor, from $\C$ to $\LSimp(\C)$. This construction is not completely similar to the tangent functor of differential geometry or tangent categories \cite{cockett_differential_2014} as we already have a notion of linearity on $\C$ given by the LNL (see def. \ref{def:PartialLinear}). Hence, the differentiation operation that will be defined needs to preserve this notion of linearity. This observation is what led us to the development of this section. \\

We call this new structure a Generalised Differential Seely Category (GDSC) as it is a generalization of DSC (see def. \ref{def:DSC}) is way that is similar to how Generalised Cartesian Differential Categories (GCDC) generalizes Cartesian Differential Categories (CDC) \cite{blute_cartesian_2009} \cite{cruttwell_cartesian_2017}. We expand more on this remark in the next section. \\

In the following, we assume that $\L$ is an additive category with product, which is necessary to have a DSC (def. \ref{def:DSC}). This makes the product into a biproduct denoted $\oplus$. In the same way as developed in the last section, this structure lifts to each fiber $\LSimp(\C)$.

\begin{prop}
\label{prop:LS_add}
If $\L$ is an additive monoidal category with biproduct, then each fiber of $\LSimp(\C)$ is an additive monoidal category with biproduct. The total category $\LSimp(\C)$ has products (as described in \ref{prop:ProductLS}) but is not additive and does not have biproducts. The reindexing functors, as well as $\Sigma$, preserve the additive structure. In each fiber, denote by $p_i$ the projections in $\LSimp(\C)$ and $p^X_i$ the projection in the fiber above $X$. Denote by $\iota^X_i$ the injection of the biproduct in each fiber.
\end{prop}

In the previous section, we pushed the idea that the fibration $\lsimp$ is a primitive notion of a trivial vector bundle fibration on $\C$. Taking this idea to heart, requiring the existence of a functor from $\C$ to $\LSimp(\C)$ should allow us to express differentiation. Let's call this functor a linear tangent functor and denote it $\Tan$.

\begin{center}
\begin{tikzcd}[column sep=4em, row sep=4em]
\Simp(\C) \arrow[dr, "\simp"'] \arrow[rr, "\F^\Simp"{name=U}, bend left=8, , shift left=0.5]  &  & \LSimp(\C)  \arrow[ll, "\U^\Simp"{name=D}, bend left=8, shift left=0.5]  \arrow[dl, "\lsimp", shift left= 1] \arrow[phantom, from=U, to=D, "\dashv" rotate=-90]  \\
& \C  \arrow[ur, "\Tan",shift left=1] & 
\end{tikzcd}
\end{center}

Asking for its existence is not enough; we require some additional axioms. Taking inspiration from the work of G. Cruttwell, J. Gallagher, JS. P. Lemay and D. Pronk \cite{cruttwell_monoidal_2022} and from geometric intuitions, it is not hard to deduce the first two axioms needed:

\begin{enumerate}[{[t.1]}]
\item $\Tan$ is a product perserving section of $\lsimp$. 
\item $\Tan(\U(A)) = (\U(A), A)$ 
\end{enumerate}

Asking that $\Tan$ is a section of $\lsimp$ forces $\Tan(X)$ to be of shape $(X, \lambda(X))$ for some $\lambda(X)$ in $\L$. This corresponds to the trivial vector bundle intuition, $(X, \lambda(X))$ is a trivial vector bundle above $X$. Since $\lsimp$ is a cartesian monoidal category (\ref{def:MonoidalFibration}, \ref{prop:ProductLS}), asking for compatibility between the product in $\C$ and $\LSimp(\C)$ seems natural as well. This means that the canonical morphisms $\varphi_{X,Y} : \Tan(X \times Y) \to \Tan(X) \with \Tan(Y)$ is an isomorphism; moreover, it is vertical. 

\begin{prop}
\label{prop:PhiVertical}
As a consequence of $\Tan$ being a section of $\lsimp$, $\varphi_{X,Y}$ is a vertical morphism above $X \times Y$.
\end{prop}

\begin{proof}
Explicitely, $\varphi_{X,Y} = \langle \Tan(\pi_1), \Tan(\pi_2) \rangle$. Since $\Tan$ is a section of $\lsimp$, for $i = 1,2$, $\Tan(\pi_i) = (\pi_i, \lambda(\pi_i))$. By definition of the product in $\LSimp(\C)$ (\ref{prop:ProductLS}) the first component of $\varphi_{X,Y}$ is $\id_{X \times Y}$ hence, $\varphi$ is vertical. 
\end{proof}

The second axiom requires that the tangent vector bundle of $\U(A)$ for $A \in \L$ is $(\U(A), A)$. This is the LL equivalent of the tangent vector bundle of a vector space $V$ being $V \times V$. In our case, the distinction $\U(A)$ and $A$ indicates that the left component is used as a point (manifold) and the right component as a vector. This is also the same as asking $\U(A)$ to be a differential object \cite{cockett_differential_2014}.\\
These two axioms are still not enough to express differentiation as in DiLL. We need a third axiom that allows us to compute the derivative of partial linear maps. In differential geometry, let $M$ be a manifold, $V, W$ two vector spaces, and $f : M \times V \to W$ a smooth map linear in its second argument (ie $\forall x \in M, \ f(x, -) :  V \to W$ is a linear map), then,

\begin{equation*}
D_{(x,a)} f(0, v) = f(x,v)
\end{equation*}
We need to express this equality in our setting. To do so, we need two other definitions:

\begin{defi}[Weakening Functor \cite{jacobs_categorical_1999}]
\label{def:WeakFunctor}
Let $W : \LSimp(\C) \to \LSimp(\C)$ be the functor defined as follows:
\begin{itemize} 
\item On objects : $W(X, A) :=  (X \times \U(A), A) =(\banana{X,A}, A)$.
\item On morphisms : let $(f,u) : (X,A) \to (Y,B)$, $W(f,u) := (\banana{f,u}, \F(\pi_1) \tensor \id_A; u)$.
\end{itemize}  
$W$ is called the weakening functor of $\LSimp(\C)$. 
\end{defi} 

This construction exists in any comprehension category \ref{def:Comprehension} and expresses the logical weakening of assumptions. Looking at its definition, we see that $\U(A)$ is added on the right-hand side and then discarded when computing $W(f,u)$. 

\begin{defi}
\label{def:I}
Define the familly of morphisms $i^{X,Y}_2 : (X \times Y, \lambda(Y)) \to (X \times Y, \lambda(X \times Y))$ as
\begin{center}
\begin{tikzcd}[column sep=4em, row sep=4em]
i^{X,Y}_2 := (X \times Y, \lambda(Y))  \arrow[r, "\iota^{X \times Y}_2"] & (X \times Y, \lambda(X) \oplus \lambda(Y)) = \Tan(X) \with \Tan(Y) \arrow[d, "\varphi^{-1}"']  \\
& (X \times Y, \lambda(X \times Y)) = \Tan(X \times Y)
\end{tikzcd}
\end{center}
Similarly, define the familly $i^{X,Y}_1 : (X \times Y, \lambda(X)) \to (X \times Y, \lambda(X \times Y))$ as $i^{X,Y}_1 := \iota^{X \times Y}_1; \varphi^{-1}$. Denote $i^{X,A}$ for $i^{X, \U(A)}$. Notice that since $\varphi$ is vectical (\ref{prop:PhiVertical}), $i^{X,Y}$ is vectical above $X \times Y$. 
\end{defi}

Notice that using axiom $t.2$ we get that $i^{X,A}_2$ is of type $(X \times \U(A), A) = W(X,A) \to \Tan \banana{X,A}$ since $t.2$ forces $\lambda(\U(A))$ to be $A$. These two familly of morphisms will allow us to express partial differential in $\LSimp(\C)$, this will lead to the third and final axiom needed. 

\begin{defi}[Partial Tangent]
\label{def:PartialTangent}
Let $f : X \times Y \to Z$ in $\C$, denote $\Tan_2(f) := i^{X, Y}_2; \Tan(f) : (X \times Y, \lambda(Y)) \to (Z, \lambda(Z)) = \Tan(Z)$ and $\Tan_1(f) := i^{X, Y}_1; \Tan(f): (X \times Y, \lambda(X)) \to (Z, \lambda(Z)) = \Tan(Z)$ the partial tangent of $f$ respectively in its second and first component. 
\end{defi}

\begin{enumerate}[{[t.3]}]
\item Partial Linearity Rule: the familly $i^{X,A}_2$ for all $X \in \C$ and $A \in \L$ defines a natural transformation from the weakening functor $W$ to $\Tan \banana{-}$. For all $(f,u) : (X,A) \to (Y,B)$ in $\LSimp(\C)$
\begin{center}
\begin{tikzcd}[column sep=4em, row sep=4em]
W(X,A) \arrow[r, "{i^{X,A}_2}"] \arrow[d, "{W(f,u)}"']  & \Tan \banana{X,A} \arrow[d, "{\Tan \banana{f,u}}"] \\
W(Y,B) \arrow[r, "{i^{Y,B}_2}"] & \Tan \banana{Y,B}
\end{tikzcd}
\end{center}
\end{enumerate}

In other words, this axiom states that for all $(f,u) : (X,A) \to (Y,B)$ in $\LSimp(\C)$, $\Tan_2 \banana{f,u} = W(f,u); i^{Y,B}_2$. Later, we show how to explicitly calculate the partial differential of a partial linear map (\ref{def:PartialLinear}). We now give the main definition of this paper.

\begin{defi}[Generalised Differential Seely Category]
\label{def:GDSC}
A Generalised Differential Seely Category (GDSC) is an additive linear-non-linear adjunction with the data of a functor $\Tan : \C \to \LSimp(\C)$, called a linear tangent functor, that satisfies the axioms:
\begin{enumerate}[{[t.1]}]
\item $\Tan$ is a product perserving section of $\lsimp$. 
\item $\Tan(\U(A)) = (\U(A), A)$ 
\item Partial Linearity Rule: the familly $i^{X,A}_2$ for all $X \in \C$ and $A \in \L$ defines a natural transformation from the weakening functor $W$ to $\Tan \banana{-}$.
\end{enumerate}
\begin{center}
\begin{tikzcd}[column sep=4em, row sep=4em]
\Simp(\C) \arrow[dr, "\simp"'] \arrow[rr, "\F^\Simp"{name=U}, bend left=8, , shift left=0.5]  &  & \LSimp(\C)  \arrow[ll, "\U^\Simp"{name=D}, bend left=8, shift left=0.5]  \arrow[dl, "\lsimp", shift left= 1] \arrow[phantom, from=U, to=D, "\dashv" rotate=-90]  \\
& \C  \arrow[ur, "\Tan",shift left=1] & 
\end{tikzcd}
\end{center}
\end{defi}

The goal of the rest of this section is to show that in any GDSC, each fiber $\LSimp(\C)_X$ is a DSC (def. \ref{def:DSC}) and that any DSC defines a GDSC. This first result implies that $\L$ is a DSC, using the isomorphism between $\LSimp(\C)_I$ and $\L$ (see remark \ref{rem:LSIisoL}). We start by showing that any DSC defines a GDSC as it is a straightforward proof and helps fix intuition. 

\begin{thm}
\label{theo_2}
Let $\L$ be a differential Seely category, then the linear-non-linear adjunction between $\L$ and $\L_\oc$, its Kleisli category, is a generalised differential Seely category.
\end{thm}

Before giving the full proof, we explain the idea behind it by describing the construction of the linear tangent functor. The linear simple category of such an LNL is described as follows:
\begin{itemize}
\item Objects: pairs $(A,B)$ with $A,B \in \L$.
\item Morphisms: a morphism from $(A,B)$ to $(A',B')$ is a pair $(f,u)$ with $f: \oc A \to A'$ in $\L$ and $u : \oc A \tensor B \to B'$ in $\L$.
\end{itemize}
To define the linear tangent functor, it suffices to define for each $f : \oc A \to B$ a morphism $u : \oc A \tensor A \to B$. Indeed, axiom t.2 forces $\Tan(A)$ to be of shape $(A,A)$ and axiom t.1 forces $\Tan(f)$ to be of shape $(f,u) : (A,A) \to (B,B)$ for some $u$. The differential of a morphism of type $\oc A \to B$ is given by postcomposition with the deriving transform $\partial_A : \oc A \tensor A \to \oc A$. The natural choice is to define $\Tan(f)$ as $(f, \partial_A; f)$.
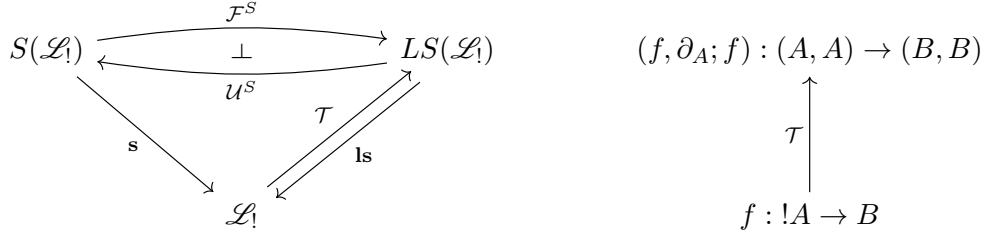
\begin{center}
\begin{tikzcd}[column sep=4em, row sep=4em]
\Simp(\L_\oc) \arrow[dr, "\simp"'] \arrow[rr, "\F^\Simp"{name=U}, bend left=8, , shift left=0.5]  &  & \LSimp(\L_\oc)  \arrow[ll, "\U^\Simp"{name=D}, bend left=8, shift left=0.5]  \arrow[dl, "\lsimp", shift left= 1] \arrow[phantom, from=U, to=D, "\dashv" rotate=-90]  & (f, \partial_A; f) : (A,A) \to (B,B) \\
& \L_\oc \arrow[ur, "\Tan",shift left=1] & & f : \oc A \to B \arrow[u, "\Tan"]
\end{tikzcd}
\end{center}
The linear rule ensures that $\Tan$ sends the identity to the identity, and the chain rule ensures that it sends composition to composition; it is indeed a functor. By construction, it is a section of $\lsimp$, and it is easy to see that it preserves products. It is straightforward, looking at its construction, to see $\Tan$ that satisfies axioms t.1 and t.2. Showing that it satisfies t.3 is more work-intensive, as it requires a lengthy computation. One must show that the following equation holds for t.3 to be valid. 
\begin{equation*}
id_{\oc (A \oplus B)} \tensor \iota_2; \partial_{A \oplus B}; \dig_{A \oplus B}; \oc (\langle \oc \pi_1, \oc \pi_2 \rangle; \id_{\oc A} \oplus \dere_B); s_{\oc A,B}; \dere_{\oc A} \tensor \dere_B = \oc \pi_1 \tensor \id_B
\end{equation*}
where $\dig_{A \oplus B}; \oc (\langle \oc \pi_1, \oc \pi_2 \rangle; \id_{\oc A} \oplus \dere_B); s_{\oc A,B}; \dere_{\oc A} \tensor \dere_B$ is the explicit description of the morphism $\eta_{A} \times \id_B; n_{\oc A, B}$ in the Kleisli category. This is not hard and only requires brutal computation. We split the computation into multiple lemmas.

\begin{lem}
\label{lem:a1}
Let $(f,u) : (A,B) \to (A',B')$ be a morphism in $\LSimp(\L_\oc)$, then 
\begin{equation*}
\banana{f,u}; \pi_2 = s_{A,B}; \id_{\oc A} \tensor \dere_B; u
\end{equation*}
\end{lem}

\begin{proof}
Let $(f ,u) : (A ,B) \to (A' ,B') \in \LSimp(\L_\oc)$. We know that $\banana{f,u}; \pi_2$ is the morphism $\eta_{A} \times \id_{B_\oc}; n_{\oc A, B}; \U(u)$ in the Kleilsi. The unit $\eta_{A}$ is simply the identity $\id_{\oc A} : \oc A \to \oc A$ and the monoidal structure of $\U$ is given by the Seely isomorphisms $s_{A,B} : \oc (A \oplus B) \to \oc A \tensor \oc B$ composed with $\dere_A \tensor \dere_B$. Let's compute $\eta_{A} \times \id_B; n_{\oc A, B}$.
\begin{align*}
\eta_{A} \times \id_B; n_{\oc A, B} &=  \dig_{A \oplus B}; \oc (\langle \oc \pi_1, \oc \pi_2 \rangle; \id_{\oc A} \oplus \dere_B); s_{\oc A,B}; \dere_{\oc A} \tensor \dere_B \\
&= \dig_{A \oplus B}; \oc \langle \oc \pi_1, \oc \pi_2 \rangle; \oc (\id_{\oc A} \oplus \dere_B); s_{\oc A,B}; \dere_{\oc A} \tensor \dere_B \\
&= \dig_{A \oplus B}; \oc \langle \oc \pi_1, \oc \pi_2 \rangle; s_{\oc A, \oc B}; \oc \id_{\oc A} \tensor \oc \dere_B; \dere_{\oc A} \tensor \dere_B \\
&= \dig_{A \oplus B}; \oc \langle \oc \pi_1, \oc \pi_2 \rangle; s_{\oc A, \oc B}; \dere_{\oc A}  \tensor \dere_{\oc B}; \id_{\oc A} \tensor \dere_B \\
&= \dig_{A \oplus B}; \oc \langle \oc \pi_1, \oc \pi_2 \rangle; \cont_{\oc A \oplus \oc B}; \oc \pi_1 \tensor \oc \pi_2; \dere_{\oc A}  \tensor \dere_{\oc B}; \id_{\oc A} \tensor \dere_B \\
&= \dig_{A \oplus B}; \cont_{\oc (A \oplus B)}; \oc \langle \oc \pi_1, \oc \pi_2 \rangle \tensor \oc \langle \oc \pi_1, \oc \pi_2 \rangle; \oc \pi_1 \tensor \oc \pi_2; \dere_{\oc A}  \tensor \dere_{\oc B}; \id_{\oc A} \tensor \dere_B \\
&= \dig_{A \oplus B}; \cont_{\oc (A \oplus B)}; \oc \oc  \pi_1 \tensor \oc \oc \pi_2; \dere_{\oc A}  \tensor \dere_{\oc B}; \id_{\oc A} \tensor \dere_B \\
&= \dig_{A \oplus B}; \cont_{\oc (A \oplus B)}; \dere_{\oc (A \oplus B)}  \tensor  \dere_{\oc (A \oplus B)}; \oc  \pi_1 \tensor \oc \pi_2; \id_{\oc A} \tensor \dere_B \\
&= \cont_{A \oplus B}; \dig_{A \oplus B} \tensor  \dig_{A \oplus B}; \dere_{\oc (A \oplus B)}  \tensor  \dere_{\oc (A \oplus B)}; \oc  \pi_1 \tensor \oc \pi_2; \id_{\oc A} \tensor \dere_B \\
&= \cont_{A \oplus B}; \oc  \pi_1 \tensor \oc \pi_2; \id_{\oc A} \tensor \dere_B \\
&= s_{A,B}; \id_{\oc A} \tensor \dere_B
\end{align*}
Hence, we obtain the desired result.
\end{proof}

To prove t.3 we need to show that the partial differential of $\eta_{A_\oc} \times \id_{B_\oc}; n_{\oc A, B}; \U(u)$ is $\pi_1^*(\id_A; u)$. Concretely, we need to show that the following holds
\begin{equation*}
\id_{\oc (A \oplus B)} \tensor \iota_2; \partial_{A \oplus B}; s_{A,B}; \id_{\oc A} \tensor \dere_B; u = \oc \pi_1 \tensor \id_B; u
\end{equation*}
We will show that the following equation holds, which implies that the required one holds. 
\begin{equation*}
\id_{\oc (A \oplus B)} \tensor \iota_2; \partial_{A \oplus B}; \id_{\oc (A \oplus B)} \tensor \iota_2; \partial_{A \oplus B}; s_{A,B}; \id_{\oc A} \tensor \dere_B = \oc \pi_1 \tensor \id_B
\end{equation*}

\begin{lem}
\label{lem:difSeely}
Let $A,B \in \L$, then the differential of the Seely isomorphism is 
\begin{equation*}
\partial_{A \oplus B}; s_{A,B} = [s_{A,B} \tensor \pi_2; \id_{\oc A} \tensor \partial_{B}] + [s_{A,B} \tensor  \pi_1; \id_{\oc A} \tensor \sigma_{\oc B, A}; \partial_{A} \tensor \id_{\oc B}]
\end{equation*}
\end{lem}

\begin{proof}
\begin{align*}
\partial_{A \oplus B}; s_{A,B} &= \partial_{A \oplus B}; \cont_{A \oplus B}; \oc \pi_1 \tensor \oc \pi_2 \\
&= \cont_{A \oplus B} \tensor \id_{A \oplus B}; \\ 
&[\id_{\oc (A \oplus B)} \tensor \partial_{A \oplus B}] + [\id_{\oc (A \oplus B)} \tensor \sigma_{\oc (A \oplus B), A \oplus B}; \partial_{A \oplus B} \tensor \id_{\oc (A \oplus B)}]; \oc \pi_1 \tensor \oc \pi_2 \\
&= \cont_{A \oplus B} \tensor \id_{A \oplus B}; [\oc \pi_1 \tensor (\partial_{A \oplus B}; \oc \pi_2)] + [\id_{\oc (A \oplus B)} \tensor \sigma_{\oc (A \oplus B), A \oplus B}; (\partial_{A \oplus B}; \oc \pi_1) \tensor \oc \pi_2] \\
&= \cont_{A \oplus B} \tensor \id_{A \oplus B}; [\oc \pi_1 \tensor (\oc \pi_2 \tensor \pi_2; \partial_{B})] + [\id_{\oc (A \oplus B)} \tensor \sigma_{\oc (A \oplus B), A \oplus B}; (\oc \pi_1 \tensor \pi_1; \partial_{A})) \tensor \oc \pi_2] \\
&= [(\cont_{A \oplus B}; \oc \pi_1 \tensor \oc \pi_2) \tensor \pi_2; \id_{\oc A} \tensor \partial_{B}] + \\
& [\cont_{A \oplus B} \tensor \id_{A \oplus B}; \id_{\oc (A \oplus B)} \tensor \sigma_{\oc (A \oplus B), A \oplus B}; (\oc \pi_1 \tensor \pi_1; \partial_{A})) \tensor \oc \pi_2] \\
&= [s_{A,B} \tensor \pi_2; \id_{\oc A} \tensor \partial_{B}] + \\
& [\cont_{A \oplus B} \tensor \id_{A \oplus B}; \id_{\oc (A \oplus B)} \tensor \sigma_{\oc (A \oplus B), A \oplus B}; (\oc \pi_1 \tensor \pi_1; \partial_{A})) \tensor \oc \pi_2] \\
&= [s_{A,B} \tensor \pi_2; \id_{\oc A} \tensor \partial_{B}] + [\cont_{A \oplus B} \tensor \id_{A \oplus B}; \oc \pi_1 \tensor \oc \pi_2  \tensor \pi_1; \id_{\oc A} \tensor \sigma_{\oc B, A}; \partial_{A} \tensor \id_{\oc B}] \\
&= [s_{A,B} \tensor \pi_2; \id_{\oc A} \tensor \partial_{B}] + [s_{A,B} \tensor  \pi_1; \id_{\oc A} \tensor \sigma_{\oc B, A}; \partial_{A} \tensor \id_{\oc B}] \qedhere
\end{align*}
\end{proof}

\begin{prop}
Let $A,B \in \L$ the the following equation is valid
\begin{equation*}
\id_{\oc (A \oplus B)} \tensor \iota_2; \partial_{A \oplus B}; s_{A,B}; \id_{\oc A} \tensor \dere_B = \oc \pi_1 \tensor \id_B
\end{equation*}
This directly implies Theorem \ref{theo_2}.
\end{prop}

\begin{proof}
\begin{align*}
\id_{\oc (A \oplus B)} &\tensor \iota_2; \partial_{A \oplus B}; s_{A,B}; \id_{\oc A} \tensor \dere_B \\
&\underset{\ref{lem:difSeely}}{=} \id_{\oc (A \oplus B)} \tensor \iota_2; [s_{A,B} \tensor \pi_2; \id_{\oc A} \tensor \partial_{B}] + [s_{A,B} \tensor  \pi_1; \id_{\oc A} \tensor \sigma_{\oc B, A}; \partial_{A} \tensor \id_{\oc B}]; \id_{\oc A} \tensor \dere_B \\
&= [\id_{\oc (A \oplus B)} \tensor \iota_2; s_{A,B} \tensor \pi_2; \id_{\oc A} \tensor \partial_{B}] + \\
& [\id_{\oc (A \oplus B)} \tensor \iota_2; s_{A,B} \tensor  \pi_1; \id_{\oc A} \tensor \sigma_{\oc B, A}; \partial_{A} \tensor \id_{\oc B}]; \id_{\oc A} \tensor \dere_B \\
&= [s_{A,B} \tensor (\iota_2;\pi_2); \id_{\oc A} \tensor \partial_{B}] + [s_{A,B} \tensor  (\iota_2;\pi_1); \id_{\oc A} \tensor \sigma_{\oc B, A}; \partial_{A} \tensor \id_{\oc B}]; \id_{\oc A} \tensor \dere_B \\
&= [s_{A,B} \tensor \id_B; \id_{\oc A} \tensor \partial_{B}] + [s_{A,B} \tensor  0; \id_{\oc A} \tensor \sigma_{\oc B, A}; \partial_{A} \tensor \id_{\oc B}]; \id_{\oc A} \tensor \dere_B \\
&= [s_{A,B} \tensor \id_B; \id_{\oc A} \tensor \partial_{B}] + 0; \id_{\oc A} \tensor \dere_B \\
&= s_{A,B} \tensor \id_B; \id_{\oc A} \tensor \partial_{B}; \id_{\oc A} \tensor \dere_B \\
&= s_{A,B} \tensor \id_B; \id_{\oc A} \tensor \weak_B \tensor \id_B \\
&= (s_{A,B}; \id_{\oc A} \tensor \weak_B) \tensor \id_B \\
&= \oc \pi_1 \tensor \id_B \qedhere
\end{align*}
\end{proof}

We now focus on the main theorem: each slice of the linear simple fibration of a GDSC is a DSC. Notice that in each fiber of $\LSimp(\L_\oc)$ we can define a deriving transform $\partial^A_B : (A, \oc B \tensor B) \to (A, \oc B)$. Applying it to a morphism $u : (A, \oc B) \to (A, B)$ yields the partial differential of $u$ in $B$. Concretely the second component of $\partial^A_B; u$ is 
 \begin{center}
\begin{tikzcd}[column sep=4em, row sep=4em]
\oc A \tensor \oc B \tensor B \arrow[r, "{\id_A \tensor \partial_B}"] & \oc A \tensor \oc B \arrow[r, "u"] & B
\end{tikzcd}
\end{center}

Our goal is to build a deriving transform in each fiber of the linear simple category of a GDSC. To achieve this, one needs to work fiberwise; hence, we define the differential of a morphism $f$, denoted $\dif(f)$, as the vertical component of $\Tan(f)$. 

\begin{defi}[Differential]
\label{def:Differential}
Let $f : X \to Y$ in $\C$, define $\dif(f)$ as the unique vertical morphism such that $\Tan(f) = \dif(f); \overline{f}_{\Tan(Y)}$ (see remark \ref{rem:FibreFactor}).  
\end{defi}

\begin{defi}[Partial Differential]
\label{def:PartialDifferential}
Let $f : X \times Y \to Z$ in $\C$, denote $\dif_2(f) := i^{X, Y}_2; \dif(f)$ and $\dif_1(f) :=  i^{X, Y}_1; \dif(f)$ the partial differentials of $f$. Since $i^{X, Y}$ is vertical (\ref{def:I}), $\dif_2(f)$ and $\dif_1(f)$ are respectively the unique vertical morphism factoring $\Tan_2(f)$ and $\Tan_1(f)$ (\ref{def:PartialTangent}). 
\end{defi}

The next few lemmas and propositions detail the behavior of the differential. We start by explaining how the differential behaves with composition and products.

\begin{prop}
\label{prop:DifferentialComp}
Let $f : X \to Y$ and $g : Y \to Z$ in $\C$ then, $\dif(f;g) = \dif(f); f^*\dif(g)$.
\end{prop}

\begin{proof}
\begin{align*}
\Tan(f);\Tan(g) &\underset{\ref{def:Differential}}{=} \dif(f); \overline{f}_{\Tan(Y)}; \dif(g); \overline{g}_{\Tan(Z)} \\
&\underset{\ref{prop:ReindexingFunctors}}{=}  \dif(f); f^*\dif(g);  \overline{f}_{\codom(\dif(g))}; \overline{g}_{\Tan(Z)}  \\
&\underset{\ref{def:SplitFib}}{=} \dif(f); f^*\dif(g);  \overline{f;g}_{\Tan(Z)} 
\end{align*}
By definition of the differential (\ref{def:Differential}), $\dif(f;g)$ is the unique vertical morphism such that $\Tan(f;g) = \dif(f;g); \overline{f;g}_{\Tan(Z)}$ hence, $\dif(f); f^*\dif(g) = \dif(f;g)$. 
\end{proof}

\begin{prop}
\label{prop:DifferentialProduct}
Let $f : X \to Y$ and $g : X \to Z$ then, $\dif( \langle f, g \rangle) =  \langle \dif(f), \dif(g) \rangle_{X};  \langle f, g \rangle^* \varphi^{-1}_{Y,Z}$.
\end{prop}

\begin{proof}
\begin{align*}
 \Tan(\langle f, g\rangle) &\underset{t.1}{=}  \langle \Tan(f), \Tan(g) \rangle; \varphi^{-1}_{Y,Z} \\
 &\underset{\ref{def:Differential}}{=}  \langle \dif(f); \overline{f}, \dif(g); \overline{g} \rangle; \varphi^{-1}_{Y,Z} \\
 &\underset{}{=} \langle \dif(f), \dif(g) \rangle_{X} ; \overline{\langle f, g \rangle};  \varphi^{-1}_{Y,Z} \\
 &\underset{\ref{prop:ReindexingFunctors}}{=} \langle \dif(f), \dif(g) \rangle_{X};  \langle f, g \rangle^* \varphi^{-1}_{Y,Z}; \overline{\langle f, g \rangle}
\end{align*}
By definition of the differential (\ref{def:Differential}), $\dif(\langle f, g\rangle)$ is the unique vertical morphism such that $\Tan(\langle f, g\rangle) = \dif(\langle f, g\rangle); \overline{\langle f, g\rangle}_{\Tan(Y \times Z)}$ hence, $\dif( \langle f, g \rangle) =  \langle \dif(f), \dif(g) \rangle_{X};  \langle f, g \rangle^* \varphi^{-1}_{Y,Z}$. 
\end{proof}

Now we turn our attention to partial differentials. To move forward, we need a few results, in particular a reformulation of axiom t.3 in terms of differential as well as a description of the differential of bilinear maps (def. \ref{def:PartialLinear}).\\
To do so, we need a lemma that expresses the partial differential of morphisms of the shape $\set{f,u} : X \times Y \to X' \times Y'$ with $(f,u) : (X, Y) \to (X',Y')$ in $\Simp(\C)$. This result will prove useful when dealing with the partial differential of a composition of morphisms.

\begin{lem}
\label{lem:PartialZero}
Let $(f,u) : (X, Y) \to (X', Y') \in \Simp(\C)$, then 
\begin{equation*}
\dif_2(\set{(f,u)}) = \\ \dif_2(\set{(f,u)}, \pi_2); \set{(f,u)}^*i^{(X', Y')}_2
\end{equation*} 
Since $u = \set{(f,u)}, \pi_2$ we can also write $\dif_2(\set{(f,u)}) = \dif_2(u); \set{(f,u)}^*i^{(X', Y')}_2$. 
\end{lem}

\begin{proof}
\begin{align*}
\dif_2(\set{(f,u)}) &\underset{\ref{def:Comprehension}}{=}  \dif_2(\langle \pi_1f, u \rangle) \\
&\underset{\ref{def:PartialDifferential}}{=} i^{X,Y}_2; \dif_2(\langle \pi_1f, u \rangle) \\
&\underset{\ref{prop:DifferentialProduct}}{=} i^{X,Y}_2;\langle \dif(\pi_1f) , \dif(u) \rangle_{X \times Y}; \set{(f,u)}^*\varphi^{-1}_{X', Y'} \\
&\underset{}{=} \langle i^{X,A}_2; \dif(\pi_1f) , \dif_2(u) \rangle_{X \times Y}; \set{(f,u)}^*\varphi^{-1}_{X', Y'} \\
&\underset{\ref{def:I}}{=} \langle \iota^{X \times Y}_2; \varphi^{-1};\dif(\pi_1); \pi_1^*\dif(f) , \dif_2(u) \rangle_{X \times Y}; \set{(f,u)}^*\varphi^{-1}_{X', Y'} \\
&\underset{t.1}{=} \langle \iota^{X \times Y}_2; \pi^{X \times Y}_1; \pi_1^*\dif(f) , \dif_2(u) \rangle_{X \times Y}; \set{(f,u)}^*\varphi^{-1}_{X', Y'} \\
&\underset{}{=} \langle 0; \pi_1^*\dif(f) , \dif_2(u) \rangle_{X \times Y}; \set{(f,u)}^*\varphi^{-1}_{X', Y'} \\
&\underset{}{=} \langle 0, \dif_2(u) \rangle_{X \times Y}; \set{(f,u)}^*\varphi^{-1}_{X', Y'} \\
&\underset{}{=}  \dif_2(u); \iota^{X \times Y}_2; \set{(f,u)}^*\varphi^{-1}_{X', Y'} \\
&\underset{}{=}  \dif_2(u); \set{(f,u)}^*\iota^{X' \times Y'}_2; \set{(f,u)}^*\varphi^{-1}_{X', Y'} \\
&\underset{\ref{def:SplitFib}}{=}  \dif_2(u); \set{(f,u)}^*(\iota^{X' \times Y'}_2; \varphi^{-1}_{X', Y'}) \\
&\underset{\ref{def:I}}{=} \dif_2(u); \set{(f,u)}^*i^{X',Y'}_2 \\
&\underset{\ref{def:Comprehension}}{=} \dif_2(\set{(f,u)}, \pi_2); \set{(f,u)}^*i^{X',Y'}_2  \qedhere
\end{align*} 
\end{proof}

We now restate the axiom t.3 in terms of the differential instead of the tangent functor.

\begin{prop}[Differential of a Partial Linear Map]
\label{prop:t3Differential}
Let $(f,u) : (X, A) \to (Y,B)$,  $\dif_2(\banana{f,u}) = \pi_1^*(\id_X, u); \banana{f,u}^*i^{Y,A}_2$
\end{prop}

\begin{proof}
\begin{align*}
\Tan_2(\banana{f,u}) &\underset{t.3}{=} W(f,u); i^{Y,A}_2\\
&\underset{\ref{def:WeakFunctor}}{=} \pi_1^*(\id_X, u);  \overline{\banana{f,u}}; i^{Y,A}_2  \\
&\underset{\ref{prop:ReindexingFunctors}}{=} \pi_1^*(\id_X, u);  \banana{f,u}^*i^{Y,A}_2;  \overline{\banana{f,u}}
\end{align*}
By unicity in the definition of the differential (\ref{def:Differential}), and by definition of the partial differential (\ref{def:PartialDifferential}), $\dif_2(\banana{f,u}) = \pi_1^*(\id_X, u); \banana{f,u}^*i^{Y,A}_2$.
 \end{proof}
 
A few more lemmas are needed to calculate the differential of bilinear maps (see definition \ref{def:LinearMap}). They describe the expected behavior of partial differentials in our fibred setting.

\begin{lem}
\label{lem:D2PartialLinear}
Let $g : X \times \U(A) \to \U(B)$ in $\C$ linear in its second argument with $g = \banana{f,u}; \pi_2$ then, $\dif_2(g) = \pi_1^*(\id_X, u)$.
\end{lem}

\begin{proof}
\begin{align*}
\dif_2(g) &\underset{\ref{def:PartialDifferential}}{=} i^{X,A}_2; \dif(g) \\
&\underset{\ref{def:LinearMap} }{=} i^{X,A}_2; \dif(\banana{f,u}; \pi_2) \\
&\underset{\ref{prop:DifferentialComp}}{=}  i^{X,A}_2; \dif(\banana{f,u}); \banana{f,u}^*\dif(\pi_2) \\
&\underset{\ref{def:PartialDifferential}}{=}  \dif_2(\banana{f,u}); \banana{f,u}^*\dif(\pi_2) \\
&\underset{\ref{prop:t3Differential}}{=}   \pi_1^*(\id_X, u); \banana{f,u}^*i^{Y,B}_2 \banana{f,u}^*\dif(\pi_2) \\
&\underset{}{=}   \pi_1^*(\id_X, u); \banana{f,u}^*(i^{Y,B}_2; \dif(\pi_2)) \\
&\underset{}{=} \pi_1^*(\id_X, u) \qedhere
\end{align*}
\end{proof}

\begin{lem}
\label{lem:D1D2}
Let $f: X \times Y \to Z$, then $\dif_1(f) = \sigma_{X,Y}^*(\dif_2(\sigma_{Y,X};f))$.
\end{lem}

\begin{proof}
\begin{align*}
\dif_1(f) &\underset{\ref{def:PartialDifferential} }{=} i^{X,Y}_1; \dif(f) \\
&\underset{}{=} i^{X,Y}_1; \dif(\sigma_{X,Y}; \sigma_{Y,X};f) \\
&\underset{\ref{prop:DifferentialComp} }{=} i^{X,Y}_1;\dif(\sigma_{X,Y}); \sigma_{X,Y}^*\dif(\sigma_{Y,X};f) \\
&\underset{ }{=} i^{X,Y}_1; \dif(\langle \pi_2, \pi_1 \rangle); \sigma_{X,Y}^*\dif(\sigma_{Y,X};f) \\
&\underset{\ref{prop:DifferentialProduct} }{=}  \langle \dif(\pi_2), \dif(\pi_1) \rangle_{X \times Y};\sigma_{X,Y}^*(\varphi^{-1}_{Y,X}); \sigma_{X,Y}^*\dif(\sigma_{Y,X};f) \\
&\underset{}{=}  \langle i^{X,Y}_1; \dif(\pi_2),i^{X,Y}_1; \dif(\pi_1) \rangle_{X \times Y};\sigma_{X,Y}^*(\varphi^{-1}_{Y,X}); \sigma_{X,Y}^*\dif(\sigma_{Y,X};f) \\
&\underset{}{=}  \langle 0 , \id_{\lambda(X)} \rangle_{X \times Y};\sigma_{X,Y}^*(\varphi^{-1}_{Y,X}); \sigma_{X,Y}^*\dif(\sigma_{Y,X};f) \\
&\underset{}{=} \sigma_{X,Y}^*(\iota^{Y,X}_2);\sigma_{X,Y}^*(\varphi^{-1}_{Y,X}); \sigma_{X,Y}^*\dif(\sigma_{Y,X};f) \\
&\underset{}{=} \sigma_{X,Y}^*(\iota^{Y,X}_2; \varphi^{-1}_{Y,X}; \dif(\sigma_{Y,X};f)) \\
&\underset{\ref{def:I}}{=} \sigma_{X,Y}^*(i^{Y,X}_2; \dif(\sigma_{Y,X};f)) \\
&\underset{\ref{def:PartialDifferential}}{=} \sigma_{X,Y}^*(\dif_2(\sigma_{Y,X};f)) \qedhere
\end{align*}
\end{proof}

\begin{cor}
\label{cor:D1PartialLinear}
Let $g : \U(A) \times X \to \U(B)$ in $\C$ linear in its first argument with $\sigma_{X,\U(A)}; g = \banana{f,u}; \pi_2$ then, $\dif_1(g) = \pi_2^*(\id_X, u)$.
\end{cor}

\begin{proof}
\begin{align*}
\dif_1(g) &\underset{\ref{lem:D1D2}}{=} \sigma_{\U(A),X}^*(\dif_2(\sigma_{X,\U(A)};g)) \\
&\underset{\ref{lem:D2PartialLinear}}{=} \sigma_{\U(A),X}^*\pi_1^*(\id_X, u) \\
&\underset{\ref{def:SplitFib}}{=}  (\sigma_{\U(A),X}; \pi_1)^*(\id_X, u) \\
&\underset{}{=} \pi_2^*(\id_X, u)  \qedhere
\end{align*}
\end{proof}

\begin{lem}
\label{lem:D=D12}
Let $f : X \times Y \to Z$ in $\C$ then $\dif(f) =\varphi_{X,Y};  (p^{X \times Y}_1; \dif_1(f)) + (p^{X \times Y}_2; \dif_2(f))$. 
\end{lem}

\begin{proof}
\begin{align*}
\dif(f) &\underset{}{=} \varphi_{X,Y}; (p^{X \times Y}_1; \iota^{X \times Y}_1) + (p^{X \times Y}_2; \iota^{X \times Y}_2); \varphi_{X,Y}^{-1}; \dif(f) \\
&\underset{\ref{def:I}}{=} \varphi_{X,Y}; (p^{X \times Y}_1; i^{X,Y}_1) + (p^{X \times Y}_2; i^{X,Y}_2); \dif(f) \\
&\underset{\ref{def:PartialDifferential}}{=}  \varphi_{X,Y}; (p^{X \times Y}_1; \dif_1(f)) + (p^{X \times Y}_2; \dif_2(f)) \qedhere
\end{align*}
\end{proof}

\begin{cor}
\label{cor:DifBi}
Let $g : \U(A) \times \U(B) \to \U(C)$ be a bilinear morphism, with $g = \banana{f,u}; \pi_2$ and $\sigma_{\U(B), \U(A)}; g = \banana{f', u'}; \pi_2$ then, 
\begin{equation*}
\dif(g) = \varphi_{\U(A),\U(B)};  (p^{\U(A) \times \U(B)}_1; \pi_2^*(\id_{\U(B)}, u')) + (p^{\U(A) \times \U(B)}_2; \pi_1^*(\id_{\U(A)}, u)).
\end{equation*}
\end{cor}

\begin{proof}
\begin{align*}
\dif(g) &\underset{\ref{lem:D=D12}}{=} \varphi_{\U(A),\U(B)};  (p^{\U(A) \times \U(B)}_1; \dif_1(g)) + (p^{\U(A) \times \U(B)}_2; \dif_2(g)) \\
&\underset{\ref{lem:D2PartialLinear}}{=} \varphi_{\U(A),\U(B)};  (p^{\U(A) \times \U(B)}_1; \dif_1(g)) + (p^{\U(A) \times \U(B)}_2; \pi_1^*(\id_{\U(A)}, u)) \\
&\underset{\ref{cor:D1PartialLinear}}{=} \varphi_{\U(A),\U(B)};  (p^{\U(A) \times \U(B)}_1; \pi_2^*(\id_{\U(B)}, u')) + (p^{\U(A) \times \U(B)}_2; \pi_1^*(\id_{\U(A)}, u)) \qedhere
\end{align*}
\end{proof}

Most of the preliminary results are done, we now turn our efforts towards constructing a deriving transform (see def. \ref{def:DSC}) in each slice of the linear simple fibration. Each fiber already has the structure of an additive Seely category; what is left is to define a deriving transform. To define it, we start by defining a family of morphisms $\D^X_Y$ that will be the base of this construction.

\begin{defi}
\label{def:D}
For every $(X, Y) \in \Simp(\C)$, define the familly of morphisms
\begin{equation*}
\D^X_Y := \dif_2(\set{\eta_{(X,Y)}}; \pi_2) : (X \times Y, \lambda(Y)) \to ( X \times Y, \F(Y))
\end{equation*}
with $\eta$ the unit of the adjunction $\F^\Simp \dashv \U^\Simp$. Denote $\D^X_A$ for $\D^X_{\U(A)}$.
\end{defi}

This family of morphisms is very close to what should be a deriving transform. In the case were $Y = \U(A)$, $\D^X_A$ is of type $(X \times \U(A), A) \to ( X \times \U(A), \oc A)$. Before building from it, the deriving transform, we expand a bit more on this family. First, exactly like with a deriving transform, composing it with a morphism of shape $\F(X) \to A$ allows us to compute the differential.

\begin{prop}
\label{prop:Dprecomp}
Let $f : (X, \F(Y)) \to (X, A)$ in $\LSimp(\C)_X$ then $\dif_2(\set{f^\dagger}; \pi_2) = \D^X_Y; \pi_1^*f$. 
\end{prop}

\begin{proof}
\begin{align*}
\dif_2(\set{f^\dagger}; \pi_2) &\underset{\ref{def:dagger}}{=} \dif_2(\set{\eta_{(X,Y)}}; \banana{f}; \pi_2) \\
&\underset{\ref{prop:DifferentialComp}}{=}  \dif_2(\set{\eta_{(X,Y)}}); \set{\eta_{(X,Y)}}^*\dif(\banana{f}; \pi_2) \\
&\underset{\ref{lem:PartialZero}}{=} \D^X_Y; \set{\eta_{(X,Y)}}^*\dif_2(\banana{f}; \pi_2) \\
&\underset{\ref{lem:D2PartialLinear}}{=}  \D^X_Y; \set{\eta_{(X,Y)}}^*\pi_1^*(f) \\
&\underset{\ref{def:SplitFib}}{=} \D^X_Y; \pi_1^*(f) \qedhere
\end{align*}
\end{proof}

This construction is a bit more general than the deriving transform of a DSC, as $\D$ allows us to express internally (without talking about the functor $\Tan$ or the category $\C$) the differential of maps of shape $\F(X) \to A$ instead of just maps of shape $\oc A \to B$. This is the result of considering any LNL instead of the Kleisli category. \\

The family $\D$ defines a natural transformation. To prove it, one needs this intermediate result that will prove quite useful later on.  

\begin{prop}
\label{prop:DDifferential}
For every $(f,u) : (X, Y) \to (X', Y')$ in $\Simp(\C)$ the following equality holds
\begin{equation*}
\dif_2(\set{(f,u)}; \pi_2); \set{(f,u)}^*\D^{X'}_{Y'} = \D^X_Y; \pi_1^*\F^\Simp (\id_X ,u)
\end{equation*}
\end{prop}

\begin{proof}
$\eta$ is a natural transform from $\id_{\Simp(\C)}$ to $\U^\Simp \F^\Simp$. Hence, $\Tan(\set{\eta})$ is a natural transform from $\Tan \set{-}$ to $\Tan \banana{ \F^\Simp - }$. For every $(f,u) : (X, Y) \to (X', Y')$ in $\Simp(\C)$ the following diagram commutes:
\begin{center}
\begin{tikzcd}[column sep=4em, row sep=4em]
(X \times Y, \lambda(X \times Y)) \arrow[r, "\Tan(\set{\eta_{(X,Y)}})"] \arrow[d, "\Tan \set{(f,u)}"]& (X \times \U \F (Y), \lambda(X \times \U \F (Y))) \arrow[d, "\Tan \banana{\F^\Simp (f,u)}"] \\
(X \times Y, \lambda(X \times Y)) \arrow[r, "\Tan(\set{\eta_{(X',Y')}})"'] & (X \times \U \F Y, \lambda(X \times \U \F (Y))) 
\end{tikzcd}
\end{center}
By looking in the fiber above $X \times Y$ we get:
\begin{equation*}
\dif(\set{\eta_{(X,Y)}}); \set{\eta_{(X,Y)}}^*\dif(\banana{\F^\Simp (f,u)}) = \dif(\set{(f,u)}); \set{(f,u)}^*\dif(\set{\eta_{(X',Y')}})
\end{equation*}
Precomposing with $i^{X,Y}_2$ we obtain: 
\begin{align*}
i^{X,Y}_2; \dif(\set{\eta_{(X,Y)}}); &\set{\eta_{(X,Y)}}^*\dif(\banana{\F^\Simp (f,u)}) \\
&\underset{\ref{def:PartialDifferential}}{=} \dif_2(\set{\eta_{(X,Y)}}); \set{\eta_{(X,Y)}}^*\dif(\banana{\F^\Simp (f,u)}) \\
&\underset{\ref{lem:PartialZero}}{=} \dif_2(\set{\eta_{(X,Y)}}; \pi_2); \set{\eta_{(X,Y)}}^* i^{X, \F(Y)}_2; \set{\eta_{(X,Y)}}^*\dif(\banana{\F^\Simp (f,u)}) \\
&\underset{\ref{def:D}}{=} \D^X_Y; \set{\eta_{(X,Y)}}^* i^{X, \F(Y)}_2; \set{\eta_{(X,Y)}}^*\dif(\banana{\F^\Simp (f,u)}) \\
&\underset{\ref{def:PartialDifferential}}{=} \D^X_Y; \set{\eta_{(X,Y)}}^*\dif_2(\banana{\F^\Simp (f,u)}) \\
&\underset{\ref{prop:t3Differential}}{=} \D^X_Y; \set{\eta_{(X,Y)}}^*(\pi_1^*\F^\Simp (\id_X ,u);  \banana{\F^\Simp(f,u)}^*i^{X',\F(Y')}_2) \\
&\underset{\ref{def:SplitFib}}{=} \D^X_Y; \pi_1^*\F^\Simp (\id_X ,u);  (\set{\eta_{(X,Y)}} \banana{\F^\Simp(f,u)})^*i^{X',\F(Y')}_2 \\
&\underset{}{=} \D^X_Y; \pi_1^*\F^\Simp (\id_X ,u);  (\set{(f,u)}; \set{\eta_{(X',Y')}})^*i^{X',\F(Y')}_2
\end{align*}
\begin{align*}
i^{X,Y}_2; \dif(\set{(f,u)}); & \set{(f,u)}^*\dif(\set{\eta_{(X',Y')}}) \\
&\underset{\ref{def:PartialDifferential}}{=} \dif_2(\set{(f,u)}); \set{(f,u)}^*\dif(\set{\eta_{(X',Y')}}) \\
&\underset{\ref{lem:PartialZero}}{=} \dif_2(\set{(f,u)}; \pi_2); \set{(f,u)}^*i^{X',Y'}_2; \set{(f,u)}^*\dif(\set{\eta_{(X',Y')}}) \\
&\underset{\ref{def:PartialDifferential}}{=} \dif_2(\set{(f,u)}; \pi_2); \set{(f,u)}^*\dif_2(\set{\eta_{(X',Y')}}) \\
&\underset{\ref{lem:PartialZero}}{=} \dif_2(\set{(f,u)}; \pi_2); \set{(f,u)}^*(\D^{X'}_{Y'}; \set{\eta_{(X',Y')}}^*i^{X',\F(Y')}_2) \\
&\underset{\ref{def:SplitFib}}{=}  \dif_2(\set{(f,u)}; \pi_2); \set{(f,u)}^*\D^{X'}_{Y'}; (\set{(f,u)}; \set{\eta_{(X',Y')}})^*i^{X',\F(Y')}_2
\end{align*}
Looking at the definition of $i$ (\ref{def:I}) we notice that it is a section (since $\iota$ is a section and $\phi$ and isomorphism) hence:
\begin{equation*}
\D^X_Y; \pi_1^*\F^\Simp (\id_X ,u) = \dif_2(\set{(f,u)}; \pi_2); \set{(f,u)}^*\D^{X'}_{Y'} \qedhere
\end{equation*}
\end{proof}

\begin{prop}
\label{prop:Dnat}
The familly of morphisms $\D^X_A$, for every $(X, A) \in \LSimp(\C)$ is a natural transformation from $W \implies \oc^\Simp W$. Diagramatically, for every $(f,u) : (X,A) \to (Y,B) \in \LSimp(\C)$:
\begin{center}
\begin{tikzcd}[column sep=4em, row sep=4em]
(X \times \U(A), A) \arrow[r, "{\D^X_A}"] \arrow[d, "{W(f,u)}"'] & (X \times \U(A), \oc A) \arrow[d, "{\oc^\Simp W(f,u)}"] \\
(Y \times \U(B), B) \arrow[r, "{\D^Y_B}"'] & (Y \times \U(B), \oc B) 
\end{tikzcd}
\end{center} 
\end{prop}

\begin{proof}
Let $(f,u) : (X,A) \to (Y,B)$ in $\LSimp(\C)$. Using proposition \ref{prop:DDifferential} we know that 
\begin{equation*}
\D^X_A; \pi_1^* \oc^\Simp (\id_X ,u) = \dif_2(\banana{f,u}; \pi_2); \banana{f,u}^*\D^{Y}_B
\end{equation*}
We postcompose both sides with $\overline{\banana{f,u}}$:
\begin{align*}
\D^X_A; \pi_1^* \oc^\Simp (\id_X ,u); \overline{\banana{f,u}} &\underset{\ref{rem:FibFuncIndex}}{=} \D^X_A;  \oc^\Simp \pi_1^*(\id_X ,u); \overline{\banana{f,u}} \\
&\underset{\ref{def:WeakFunctor}}{=}  \D^X_A;  \oc^\Simp W(f,u)
\end{align*}
\begin{align*}
 \dif_2(\banana{f,u}; \pi_2); \banana{f,u}^*\D^{Y}_B; \overline{\banana{f,u}} &\underset{\ref{prop:ReindexingFunctors}}{=} \dif_2(\banana{f,u}; \pi_2); \overline{\banana{f,u}}; \D^{Y}_B\\
&\underset{\ref{prop:DifferentialComp}}{=}  \dif_2(\banana{f,u}); \banana{f,u}^*D(\pi_2); \overline{\banana{f,u}}; \D^{Y}_B \\
&\underset{\ref{prop:t3Differential}}{=}  \pi_1^*(\id_X, u); \banana{f,u}^*i^{Y,A}_2; \banana{f,u}^*D(\pi_2); \overline{\banana{f,u}}; \D^{Y}_B \\
&\underset{}{=}  \pi_1^*(\id_X, u); \banana{f,u}^*(i^{Y,A}_2; D(\pi_2)); \overline{\banana{f,u}}; \D^{Y}_B \\
&\underset{}{=} \pi_1^*(\id_X, u); \overline{\banana{f,u}}; \D^{Y}_B \\
&\underset{\ref{def:WeakFunctor}}{=} W(f,u); \D^{Y}_B
\end{align*}
The family of morphisms $\D^X_A$ is indeed a natural transformation from $W$ to $\oc^\Simp W$. 
\end{proof}

Notice that we restricted our setting to the familly $(\D^X_A)_{X \in \C, A \in \L}$ istead of the familly $(\D^X_Y)_{X,Y \in \C}$. This last family is also a natural transformation relating morphisms in $\Simp(\C)$ to their partial differential in $\LSimp(\C)$. 

\begin{defi}[Deriving transform]
\label{def:partial}
For every $(X,A) \in \LSimp(\C)$, define \\ $\partial^X_A : \oc^\Simp (X, A) \tensor_X (X,A) \to \oc^\Simp (X,A)$ as 
\begin{equation*}
\partial^X_A := \Sigma_{(X,A)} \D^X_A; \mu^{(X,A)}_{(X, ! A)}
\end{equation*}
Where $\mu$ is the counit of the coproduct adjunction (see prop. \ref{linear_sigma_types}). It defines a family of morphisms $\partial$ called a deriving transform. 
\end{defi}

In other words, for every $X \in \C$ and every $A \in \L$, $\partial^X_A$ is of type $(X, \oc A \tensor A) \to (X, \oc A)$. Now we work towards proving that for every $X \in \C$ fixed, the family $(\partial^X_A)_{A \in \L}$ is a deriving transform (def. \ref{def:DSC}) in the fiber $\LSimp(\C)_X$. 

\begin{prop}
For every $X \in \C$, the familly of morphisms $\partial^X_A : \oc^\Simp (X, A) \tensor_X (X,A) \to \oc^\Simp (X,A)$ is a natural transformation in $\LSimp(\C)_X$ from $\oc^\Simp_X \tensor_X \id_{\LSimp(\C)_X}$ to $\oc^\Simp_X$ where $\oc^\Simp_X$ denotes the restriction of the functor $\oc^\Simp$ to the fiber above $X$. 
\end{prop}

\begin{proof}
Let $f : (X, A) \to (X, B)$ in $\LSimp(\C)_X$. 
\begin{align*}
\oc^\Simp f \tensor_X f; \partial^X_B &\underset{\ref{rem:SigmaContWeak}}{=} \oc^\Simp f \tensor_X \id_A; \Sigma_{X,A} \pi_1^*f; \partial^X_B \\
&\underset{\ref{def:partial}}{=}  \oc^\Simp f \tensor_X \id_A; \Sigma_{(X,B)} (\pi_1^*f; \D^X_B); \mu^{(X,B)}_{(X, ! B)}\\
&\underset{\ref{lem:SigmaF}}{=}  \Sigma_{(X,A)} [\banana{f}^*(\pi_1^*f; \D^X_B)]; \mu^{(X,A)}_{(X, ! B)} \\
&\underset{\ref{def:SplitFib}}{=}  \Sigma_{(X,A)} [\pi_1^*f; \banana{f}^*\D^X_B)]; \mu^{(X,A)}_{(X, ! B)} \\
&\underset{\ref{lem:D2PartialLinear}}{=} \Sigma_{(X,A)} [\dif_2(\banana{f};\pi_2); \banana{f}^*\D^X_B)]; \mu^{(X,A)}_{(X, ! B)} \\
&\underset{\ref{def:D}}{=} \Sigma_{(X,A)} [\dif_2(\banana{f};\pi_2); \banana{f}^*\dif_2(\set{\eta_{(X,B)}}; \pi_2))]; \mu^{(X,A)}_{(X, ! B)} \\
&\underset{\ref{lem:PartialZero}}{=} \Sigma_{(X,A)} [\dif_2(\banana{f}); \banana{f}^*\dif(\set{\eta_{(X,B)}}; \pi_2))]; \mu^{(X,A)}_{(X, ! B)} \\
&\underset{\ref{prop:DifferentialComp}}{=} \Sigma_{(X,A)} \dif_2(\banana{f};\set{\eta_{(X,B)}}; \pi_2)); \mu^{(X,A)}_{(X, ! B)} \\
&\underset{}{=} \Sigma_{(X,A)} \dif_2(\set{\eta_{(X,A)}}; \banana{\oc^\Simp f}; \pi_2)); \mu^{(X,A)}_{(X, ! B)} \\
&\underset{\ref{prop:DifferentialComp}}{=} \Sigma_{(X,A)} [\dif_2(\set{\eta_{(X,A)}}); \set{\eta_{(X,A)}}^*\dif(\banana{\oc^\Simp f}; \pi_2))]; \mu^{(X,A)}_{(X, ! B)} \\
&\underset{\ref{lem:PartialZero}}{=} \Sigma_{(X,A)} [\D^X_A; \set{\eta_{(X,A)}}^*\dif_2(\banana{\oc^\Simp f}; \pi_2))]; \mu^{(X,A)}_{(X, ! B)} \\
&\underset{\ref{lem:D2PartialLinear}}{=} \Sigma_{(X,A)} [\D^X_A; \set{\eta_{(X,A)}}^*\pi_1^*(\oc^\Simp f)]; \mu^{(X,A)}_{(X, ! B)} \\
&\underset{\ref{def:SplitFib}}{=} \Sigma_{(X,A)} [\D^X_A; \pi_1^*(\oc^\Simp f)]; \mu^{(X,A)}_{(X, ! B)} \\
&\underset{\ref{def:SplitFib}}{=} \Sigma_{(X,A)} \D^X_A; \Sigma_{(X,A)} \pi_1^*(\oc^\Simp f); \mu^{(X,A)}_{(X, ! B)} \\
&\underset{}{=} \Sigma_{(X,A)} \D^X_A;  \mu^{(X,A)}_{(X, ! A)}; \oc^\Simp f \\
&\underset{\ref{def:partial}}{=} \partial^X_A; \oc^\Simp f \qedhere
\end{align*}
\end{proof}

Now, exactly like with $\D$, one needs a way to compute the composition of the deriving transform with a morphism of shape $\oc A \to B$.

\begin{lem}
\label{lem:partial_dif}
Let $f : (X, \oc A) \to (X,B)$ in $\LSimp(\C)_X$, then 
\begin{equation*}
\Sigma_{(X,A)} \dif_2(\set{f^\dagger}; \pi_2) ; \mu^{(X,A)}_{(X,B)} = \partial_{(X,A)}; f
\end{equation*}
\end{lem}

\begin{proof}
\begin{align*}
\Sigma_{(X,A)} (\dif_2(\set{f^\dagger}; \pi_2)) ; \mu^{(X,A)}_{(X,B)} 
&\underset{\ref{prop:Dprecomp}}{=}  \Sigma_{(X,A)} (\D^X_A; \pi_1^*f) ; \mu^{(X,A)}_{(X,B)}\\
&\underset{}{=}  \Sigma_{(X,A)} \D^X_A; \Sigma_{(X,A)} \pi_1^*f ; \mu^{(X,A)}_{(X,B)}\\
&\underset{}{=} \Sigma_{(X,A)} \D^X_A; \mu^{(X,A)}_{(X,\oc A)}; f \\
&\underset{\ref{def:partial}}{=} \partial^X_A; f \qedhere
\end{align*}
\end{proof}

\begin{rem}
\label{rem:genpartial}
Exactly as per the last remark, we could have defined a more general version of $\partial^X$, defining for every $X,Y \in \C$, $\partial^X_Y$ as $\Sigma_{(X,Y)} \D^X_Y; \mu^{(X,Y)}_{(X, \F(Y))}$ of type $(X, \F(Y) \tensor \lambda(Y)) \to (X, \F(Y))$. For simplicity, we stick to the usual version of the deriving transform and not this more general form. 
\end{rem}

Now that we know that $\partial^X$ is a natural transform, to show that it is indeed a deriving transform, it suffices to show that it satisfies the linear rule, the chain rule, and the Leibniz rule \cite{blute_differential_2020}.

\begin{prop}[Linear Rule]
\label{prop:linearrule}
For every $X \in \C$, $\partial^X$ satisfies the linear rule.
\end{prop}

\begin{proof}
\begin{align*}
\partial^X_A; \dere^X_A &\underset{\ref{lem:partial_dif}}{=} \Sigma_{(X,A)} \dif_2(\set{(\dere^X_A)^\dagger}; \pi_2) ; \mu^{(X,A)}_{(X,A)} \\
&\underset{}{=} \Sigma_{(X,A)} \dif_2(\pi_2) ; \mu^{(X,A)}_{(X,A)} \\
&\underset{}{=} \mu^{(X,A)}_{(X,A)} \\
&\underset{\ref{rem:SigmaContWeak}}{=} \weak^X_A \tensor_X \id_{(X,A)} \qedhere
\end{align*}
\end{proof}

To show that $\partial$ satisfies the chain rule, we need a lemma. This result is simply a matter of using the comonad identities of $\Sigma \pi_1$.

\begin{lem}
\label{lem:SD=partial}
For every $X \in \C$ and $A \in \L$, $\Sigma_{(X,A)} \D^X_A = \cont^X_A\tensor_X \id_{(X,A)}; \id_{(X, \oc A)} \tensor_X \partial^X_A$
\end{lem}

\begin{proof}
\begin{align*}
\Sigma_{(X,A)} \D^X_A &\underset{}{=} \delta^{(X,A)}_{(X,A)};  \Sigma_{(X,A)} \pi_1^*(\Sigma_{(X,A)} \D^X_A); \Sigma_{(X,A)} \pi_1^*\mu^{(X,A)}_{(X,\oc A)} \\
&\underset{\ref{def:partial}}{=} \delta^{(X,A)}_{(X,A)};  \Sigma_{(X,A)} \pi_1^*(\partial^X_A) \\
&\underset{\ref{rem:SigmaContWeak}}{=} \cont^X_A\tensor_X \id_{(X,A)}; \Sigma_{(X,A)} \pi_1^*(\partial^X_A) \\
&\underset{\ref{rem:SigmaContWeak}}{=} \cont^X_A\tensor_X \id_{(X,A)}; \id_{(X, \oc A)} \tensor_X \partial^X_A \qedhere
\end{align*}
\end{proof}

\begin{prop}[Chain Rule]
\label{prop:lchainrule}
For every $X \in \C$, $\partial^X$ satisfies the chain rule.
\end{prop}

\begin{proof}
\begin{align*}
\partial^X_A; \dig^X_A &\underset{\ref{lem:partial_dif}}{=} \Sigma_{(X,A)} \dif_2(\set{(\dig_A^X)^\dagger}; \pi_2) ; \mu^{(X,A)}_{(X,\oc \oc A)} \\
&\underset{\ref{def:dagger}}{=} \Sigma_{(X,A)} \dif_2(\set{\eta_{(X,A)}; \U^\Simp(\dig^X_A)}; \pi_2) ; \mu^{(X,A)}_{(X,\oc \oc A)} \\
&\underset{}{=} \Sigma_{(X,A)} \dif_2(\set{\eta_{(X,A)}; \U^\Simp( \F^\Simp \eta_{(X, A)})}; \pi_2) ; \mu^{(X,A)}_{(X,\oc \oc A)} \\
&\underset{}{=} \Sigma_{(X,A)} \dif_2(\set{\eta_{(X,A)}; \eta_{(X, \oc A)}}; \pi_2) ; \mu^{(X,A)}_{(X,\oc \oc A)} \\
&\underset{\ref{prop:DifferentialComp}}{=} \Sigma_{(X,A)} (\D^X_A; \set{\eta_{(X, A)}}^*\dif(\eta_{(X, \oc A)}; \pi_2)) ; \mu^{(X,A)}_{(X,\oc \oc A)} \\
&\underset{\ref{lem:PartialZero}}{=} \Sigma_{(X,A)} (\D^X_A; \set{\eta_{(X, A)}}^*\D^X_{\oc A}) ; \mu^{(X,A)}_{(X,\oc \oc A)} \\
&\underset{}{=} \Sigma_{(X,A)} \D^X_A;  \Sigma_{(X,A)} (\set{\eta_{(X, A)}}^*\D^X_{\oc A}) ; \mu^{(X,A)}_{(X,\oc \oc A)} \\
&\underset{\ref{lem:SD=partial}}{=}\cont^X_A\tensor_X \id_{(X,A)}; \id_{(X, \oc A)} \tensor_X \partial^X_A;  \Sigma_{(X,A)} (\set{\eta_{(X, A)}}^*\D^X_{\oc A}) ; \mu^{(X,A)}_{(X,\oc \oc A)} \\
&\underset{\ref{lem:SigmaF}}{=} \cont^X_A\tensor_X \id_{(X,A)}; \F(\eta_{(X, A)}) \tensor_X \partial^X_A;  \Sigma_{(X, \oc A)} \D^X_{\oc A} ; \mu^{(X,A)}_{(X,\oc \oc A)} \\
&\underset{}{=} \cont^X_A\tensor_X \id_{(X,A)}; \dig^X_A \tensor_X \partial^X_A;  \Sigma_{(X, \oc A)} \D^X_{\oc A} ; \mu^{(X,A)}_{(X,\oc \oc A)} \\
&\underset{\ref{def:partial}}{=} \cont^X_A\tensor_X \id_{(X,A)}; \dig^X_A \tensor_X \partial^X_A;  \partial^X_{\oc A} \qedhere
\end{align*}
\end{proof}

What is left is to show that the Leibniz rule holds in all fibers. The proof plans to exploit the fact that $(\cont^X_A)^\dagger = \eta_{(X,A)}; \Delta_{\U(\oc A)}; n_{\oc A, \oc A}$ (lem. \ref{cont_dagger}). If we know how to compute the second differential of each of these morphisms, then we will be able to compute $\partial^X_A; \cont^X_A$. We already know how to compute the differential of $\eta$, and the differential of the diagonal can be computed using the compatibility with the product. To compute the differential of $n$ we rely on the fact that it is bilinear (def. \ref{def:PartialLinear}). 

\begin{prop}
For all $A,B \in \L$, the morphism $n_{A,B} : \U(A) \times \U(B) \to \U(A \tensor B)$ is bilinear with $n_{A,B} = \banana{\id_{\U(A)}, \dere_A \tensor \id_B}; \pi_2$ and $\sigma_{\U(B), \U(A)}; n_{A,B} = \banana{\id_{\U(B)}, \sigma_{\oc B, A}; \id_A \tensor \dere_B}; \pi_2$. 
\end{prop}

\begin{proof}
We only show the first equality; the second one is done similarly.
\begin{align*}
\banana{\id_{\U(A)}, \dere_A \tensor \id_B}; \pi_2 
&\underset{\ref{def:Comprehension}}{=} \eta_{\U(A)} \times \id_{\U(B)}; n_{\oc A, B}; \U( \dere_A \tensor \id_B) \\
&\underset{}{=} \eta_{\U(A)} \times \id_{\U(B)}; \U( \dere_A) \times \id_{\U(B)}; n_{A, B} \\
&\underset{}{=} (\eta_{\U(A)};U( \dere_A))  \times \id_{\U(B)}; n_{A, B} \\
&\underset{}{=} n_{A, B} \qedhere
\end{align*}
\end{proof}

For clarity, we split the computation of $\partial^X_A; \cont^X_A$ into multiple lemmas.

\begin{lem}
\label{lem:trivial}
Let $u : (X,Y) \to (X,Z)$ be a morphism in $\Simp(\C)_X$ of shape $\pi_2; f$ with $f : Y \to Z$ in $\C$ then, $\dif_2(\set{u}; \pi_2) = \pi_2^*\dif(f)$. 
\end{lem}

\begin{proof}
\begin{align*}
\dif_2(\set{u}; \pi_2) \underset{}{=} \dif_2( \pi_2; f) \underset{\ref{prop:DifferentialComp}}{=} \dif_2( \pi_2); \pi_2^*\dif(f) \underset{}{=} \pi_2^*\dif(f) 
\end{align*}
\end{proof}

\begin{lem}
\label{lem:Leib1}
For every $X \in \C$ and $A \in \L$, 
\begin{equation*}
\dif_2(\set{\eta_{(X,A)}; \Delta^X_{\U(\oc A}}; \pi_2) = \D^X_A; (\Delta^{\oplus})^{X \times \U(A)}_{\oc A};  (\set{\eta_{(X,A)};\Delta^X_{\U(\oc A)}}; \pi_2)^*\varphi^{-1}_{\oc A, \oc A}
\end{equation*}
\end{lem}

\begin{proof}
\begin{align*}
\dif_2(\set{\eta_{(X,A)}; &\Delta^X_{\U(\oc A}}; \pi_2) \\
&\underset{\ref{prop:DifferentialComp}}{=} \dif_2(\set{\eta_{(X,A)}}); \set{\eta_{(X,A)}}^*\dif(\set{\Delta^X_{\U(\oc A)}}; \pi_2) \\
&\underset{\ref{lem:PartialZero}}{=}  \D^X_A; \set{\eta_{(X,A)}}^*\dif_2(\set{\Delta^X_{\U(\oc A)}}; \pi_2) \\
&\underset{\ref{lem:trivial}}{=} \D^X_A; \set{\eta_{(X,A)}}^*(\pi_2^*\dif(\Delta_{\U(\oc A)})) \\
&\underset{\ref{def:SplitFib}}{=} \D^X_A; (\set{\eta_{(X,A)}}; \pi_2)^*\dif(\Delta_{\U(\oc A)})) \\
&\underset{}{=} \D^X_A; (\set{\eta_{(X,A)}}; \pi_2)^*\dif(\Delta_{\U(\oc A)})) \\
&\underset{}{=} \D^X_A; (\set{\eta_{(X,A)}}; \pi_2)^*\dif(\langle \id_{\U(\oc A)}, \id_{\U(\oc A)} \rangle) \\
&\underset{\ref{prop:DifferentialProduct}}{=} \D^X_A; (\set{\eta_{(X,A)}}; \pi_2)^*[ \langle \dif(\id_{\U(\oc A)}), \dif(\id_{\U(\oc A)}) \rangle_{\U(\oc A)}; \Delta_{\U(\oc A)}^* \varphi^{-1}_{\U(\oc A), \U(\oc A)}] \\
&\underset{}{=} \D^X_A; (\set{\eta_{(X,A)}}; \pi_2)^*[ \langle \id_{\U(\oc A)}, \id_{\U(\oc A)} \rangle_{\U(\oc A)}; \Delta_{\U(\oc A)}^* \varphi^{-1}_{\U(\oc A), \U(\oc A)}] \\
&\underset{}{=} \D^X_A; (\set{\eta_{(X,A)}}; \pi_2)^*[(\Delta^{\oplus})^{\U(\oc A)}_{\oc A}; \Delta_{\U(\oc A)}^* \varphi^{-1}_{\U(\oc A), \U(\oc A)}] \\
&\underset{\ref{def:SplitFib}}{=} \D^X_A; (\set{\eta_{(X,A)}}; \pi_2)^*(\Delta^{\oplus})^{\U(\oc A)}_{\oc A};  (\set{\eta_{(X,A)}}; \pi_2; \Delta_{\U(\oc A)})^* \varphi^{-1}_{\U(\oc A), \U(\oc A)} \\
&\underset{}{=} \D^X_A; (\Delta^{\oplus})^{X \times \U(A)}_{\oc A};  (\set{\eta_{(X,A)}}; \pi_2; \Delta_{\U(\oc A)})^*\varphi^{-1}_{\U(\oc A), \U(\oc A)} \\
&\underset{}{=} \D^X_A; (\Delta^{\oplus})^{X \times \U(A)}_{\oc A};  (\set{\eta_{(X,A)};\Delta^X_{\U(\oc A)}}; \pi_2)^*\varphi^{-1}_{\U(\oc A), \U(\oc A)} \qedhere
\end{align*}
\end{proof}

\begin{lem}
\label{lem:Leib2}
For all $X \in \C$ and $A \in \L$ the following equality holds
\begin{equation*}
\dif_2(\set{(\cont^X_A)^\dagger}; \pi_2) = [\D^X_A; \pi_2^*(\id_{\U(A)}, \sigma_{\oc A, \oc A}) ] + [\D^X_A; \pi_2^*(\id_{\U(A)}, \id_{\oc A} \tensor \id_{\oc A}))]
\end{equation*}
\end{lem}

\begin{proof}
Let $X \in \C$ and $A \in \L$.
\begin{align*}
\dif_2(\set{(\cont^X_A)^\dagger}; \pi_2) 
&\underset{\ref{cont_dagger}}{=} \dif_2(\set{\eta_{(X,A)}; \Delta^X_{\U(\oc A)}; n^X_{\oc^\Simp A, \oc^\Simp A}}; \pi_2) \\
&\underset{}{=} \dif_2(\set{\eta_{(X,A)}; \Delta^X_{\U(\oc A)}}; \set{n^X_{\oc^\Simp A, \oc^\Simp A}}; \pi_2) \\
&\underset{\ref{prop:DifferentialComp}}{=} \dif_2(\set{\eta_{(X,A)}; \Delta^X_{\U(\oc A)}}); \set{\eta_{(X,A)}; \Delta^X_{\U(\oc A)}}^*\dif(\set{n^X_{\oc^\Simp A, \oc^\Simp A}}; \pi_2) \\
&\underset{\ref{lem:PartialZero}}{=} \dif_2(\set{\eta_{(X,A)}; \Delta^X_{\U(\oc A)}}; \pi_2); \set{\eta_{(X,A)}; \Delta^X_{\U(\oc A)}}^*\dif_2(\set{n^X_{\oc^\Simp A, \oc^\Simp A}}; \pi_2) \\
&\underset{\ref{lem:trivial}}{=} \dif_2(\set{\eta_{(X,A)}; \Delta^X_{\U(\oc A)}}; \pi_2); \set{\eta_{(X,A)}; \Delta^X_{\U(\oc A)}}^*(\pi_2^*\dif(n_{\oc A, \oc A})) \\
&\underset{\ref{def:SplitFib}}{=} \dif_2(\set{\eta_{(X,A)}; \Delta^X_{\U(\oc A)}}; \pi_2); (\set{\eta_{(X,A)}; \Delta^X_{\U(\oc A)}}; \pi_2)^*\dif(n_{\oc A, \oc A}) \\
&\underset{\ref{cor:DifBi}}{=} \dif_2(\set{\eta_{(X,A)}; \Delta^X_{\U(\oc A)}}; \pi_2); (\set{\eta_{(X,A)}; \Delta^X_{\U(\oc A)}}; \pi_2)^*[\varphi_{\U(\oc A),\U(\oc A)}; \\ 
&(p^{\U(\oc A) \times \U(\oc A)}_1; \pi_2^*(\id_{\U(\oc A)}, \sigma_{\oc \oc A, \oc A}; \id_{\oc A} \tensor \dere_{\oc A})) + (p^{\U(\oc A) \times \U(\oc A)}_2; \pi_1^*(\id_{\U(\oc A)}, \dere_{\oc A} \tensor \id_{\oc A}))] \\
&\underset{\ref{lem:Leib1}}{=} \D^X_A; (\Delta^{\oplus})^{X \times \U(A)}_{\oc A}; (\set{\eta_{(X,A)}; \Delta^X_{\U(\oc A)}}; \pi_2)^*[ \\ 
&(p^{\U(\oc A) \times \U(\oc A)}_1; \pi_2^*(\id_{\U(\oc A)}, \sigma_{\oc \oc A, \oc A}; \id_{\oc A} \tensor \dere_{\oc A})) + (p^{\U(\oc A) \times \U(\oc A)}_2; \pi_1^*(\id_{\U(\oc A)}, \dere_{\oc A} \tensor \id_{\oc A}))] \\
&\underset{}{=} [\D^X_A; (\Delta^{\oplus})^{X \times \U(A)}_{\oc A}; p^{X \times \U(A) }_1; (\set{\eta_{(X,A)}; \Delta^X_{\U(\oc A)}}; \pi_2)^*(\pi_2^*(\id_{\U(\oc A)}, \sigma_{\oc \oc A, \oc A}; \id_{\oc A} \tensor \dere_{\oc A}))] + \\ 
&[\D^X_A; (\Delta^{\oplus})^{X \times \U(A)}_{\oc A}; p^{X \times \U(A) }_2; (\set{\eta_{(X,A)}; \Delta^X_{\U(\oc A)}}; \pi_2)^*(\pi_1^*(\id_{\U(\oc A)}, \dere_{\oc A} \tensor \id_{\oc A}))] \\
&\underset{}{=} [\D^X_A; (\set{\eta_{(X,A)}; \Delta^X_{\U(\oc A)}}; \pi_2)^*(\pi_2^*(\id_{\U(\oc A)}, \sigma_{\oc \oc A, \oc A}; \id_{\oc A} \tensor \dere_{\oc A}))] + \\ 
&[\D^X_A; (\set{\eta_{(X,A)}; \Delta^X_{\U(\oc A)}}; \pi_2)^*(\pi_1^*(\id_{\U(\oc A)}, \dere_{\oc A} \tensor \id_{\oc A}))] \\
&\underset{\ref{def:SplitFib}}{=} [\D^X_A; (\set{\eta_{(X,A)}; \Delta^X_{\U(\oc A)}}; \pi_2; \pi_2)^*(\id_{\U(\oc A)}, \sigma_{\oc \oc A, \oc A}; \id_{\oc A} \tensor \dere_{\oc A})] + \\ 
&[\D^X_A; (\set{\eta_{(X,A)}; \Delta^X_{\U(\oc A)}}; \pi_2; \pi_1)^*(\id_{\U(\oc A)}, \dere_{\oc A} \tensor \id_{\oc A})] \\
&\underset{\ref{def:SplitFib}}{=} [\D^X_A; (\pi_2; \eta_{\U(A)})^*(\id_{\U(\oc A)}, \sigma_{\oc \oc A, \oc A}; \id_{\oc A} \tensor \dere_{\oc A})] + \\ 
&[\D^X_A; (\pi_2; \eta_{\U(A)})^*(\id_{\U(\oc A)}, \dere_{\oc A} \tensor \id_{\oc A})] \\
&\underset{}{=}[\D^X_A; \pi_2^*(\id_{\U( A)}, \sigma_{\oc A, \oc A}; \id_{\oc A} \tensor \id_{\oc A})) ] + [\D^X_A; \pi_2^*(\id_{\U( A)}, \id_{\oc A} \tensor \id_{\oc A})] \\
&\underset{}{=}[\D^X_A; \pi_2^*(\id_{\U(A)}, \sigma_{\oc A, \oc A}) ] + [\D^X_A; \pi_2^*(\id_{\U(A)}, \id_{\oc A} \tensor \id_{\oc A}))] \qedhere
\end{align*}
\end{proof}

\begin{lem}
\label{lem:Leib3}
For all $X \in \C$ and $A \in \L$ the following equalities holds
\begin{equation*}
\Sigma_{(X,A)} \pi_2^*(\id_{\U( A)}, \sigma_{\oc A, \oc A}) = \weak_X \tensor \cont_A \tensor \id_{\oc A}; \id_{\oc A} \tensor \sigma_{\oc A, \oc A}
\end{equation*}
\begin{equation*}
\Sigma_{(X,A)} \pi_2^*(\id_{\U(A)}, \id_{\oc A} \tensor \id_{\oc A} ) = \cont^X_A \tensor_X \id_{\oc^\Simp (X,A)}
\end{equation*}
\end{lem}

\begin{proof}
Let's start by computing the second component of $\Sigma_{(X,A)} \pi_2^*(\id_{\U( A)}, \sigma_{\oc A, \oc A})$.
\begin{align*}
&(\Sigma_{(X,A)} \pi_2^*(\id_{\U( A)}, \sigma_{\oc A, \oc A}))_2 \\
&\underset{\ref{linear_sigma_types}}{=} \id_{\F(X)} \tensor \cont_A \tensor \id_{\oc A}; \sigma_{\F(X), \oc A} \tensor \id_{\oc A \tensor \oc A}; \id_{\oc A} \tensor [ (m_{X, \oc A}; \F(\pi_2)) \tensor \id_{\oc A}; \sigma_{\oc A, \oc A}] \\
&\underset{}{=} \id_{\F(X)} \tensor \cont_A \tensor \id_{\oc A}; \sigma_{\F(X), \oc A} \tensor \id_{\oc A \tensor \oc A}; \id_{\oc A} \tensor [ \weak_X \tensor \sigma_{\oc A, \oc A}] \\
&\underset{}{=} \weak_X \tensor \cont_A \tensor \id_{\oc A}; \id_{\oc A} \tensor \sigma_{\oc A, \oc A}
\end{align*}
From this, we easely deduce that 
\begin{equation*}
\Sigma_{(X,A)} \pi_2^*(\id_{\U(\oc A)}, \sigma_{\oc A, \oc A}) = \cont^X_A \tensor_X \id_{\oc^\Simp (X,A)}; \id_{\oc^\Simp (X,A)} \tensor_X \sigma^X_{\oc^\Simp (X,A), \oc^\Simp (X,A)}
\end{equation*}
Similarly, 
\begin{equation*}
\Sigma_{(X,A)} \pi_2^*(\id_{\U(A)}, \id_{\oc A} \tensor \id_{\oc A} ) = \cont^X_A \tensor_X \id_{\oc^\Simp (X,A)}
\end{equation*}
\end{proof}

\begin{prop}[Leibniz Rule]
For every $X \in \C$, $\partial^X$ satisfies the Leibniz rule.
\end{prop}

\begin{proof}
Let $X \in \C$ and $A \in \L$.
\begin{align*}
\partial^X_A; \cont^X_A 
&\underset{\ref{lem:partial_dif}}{=} \Sigma_{(X,A)} \dif_2(\set{\eta_{(X,A)}; \Delta^X_{\U(\oc A)}}; \pi_2); \mu^{(X,A)}_{(X, \oc A \tensor \oc A)} \\
&\underset{\ref{lem:Leib2}}{=} \Sigma_{(X,A)} ([\D^X_A; \pi_2^*(\id_{\U(A)}, \sigma_{\oc A, \oc A}) ] + [\D^X_A; \pi_2^*(\id_{\U(A)}, \id_{\oc A} \tensor \id_{\oc A})]); \mu^{(X,A)}_{(X, \oc A \tensor \oc A)} \\
&\underset{\ref{prop:LS_add}}{=} \Sigma_{(X,A)} \D^X_A; ([\Sigma_{(X,A)} \pi_2^*(\id_{\U(A)}, \sigma_{\oc A, \oc A})] + [\Sigma_{(X,A)} \pi_2^*(\id_{\U(A)}, \id_{\oc A} \tensor \id_{\oc A})]); \mu^{(X,A)}_{(X, \oc A \tensor \oc A)} \\
&\underset{\ref{lem:Leib3}}{=} \Sigma_{(X,A)} \D^X_A; ([\cont^X_A \tensor_X \id_{\oc^\Simp (X,A)}; \id_{\oc^\Simp (X,A)} \tensor_X \sigma^X_{\oc^\Simp (X,A), \oc^\Simp (X,A)}] + [\cont^X_A \tensor_X \id_{\oc^\Simp (X,A)}]); \\ & \mu^{(X,A)}_{(X, \oc A \tensor \oc A)} \\
&\underset{\ref{rem:SigmaContWeak}}{=} \Sigma_{(X,A)} \D^X_A; (\sigma^X_{\oc^\Simp (X,A), \oc^\Simp (X,A)} + \id_{\oc^\Simp (X,A) \tensor_X \oc^\Simp (X,A)}) \\
&\underset{\ref{lem:SD=partial}}{=} \cont^X_A \tensor_X \id_{(X,A)}; \id_{(X, \oc A)} \tensor_X \partial^X_A; (\sigma^X_{ (X,\oc A), (X, \oc A)} + \id_{ (X, \oc A) \tensor_X (X, \oc A)}) \\
&\underset{}{=} \cont^X_A \tensor_X \id_{(X,A)}; [ (\id_{(X, \oc A)} \tensor_X \partial^X_A;\sigma^X_{ (X,\oc A), (X, \oc A)}) + (\id_{(X, \oc A)} \tensor_X \partial^X_A)] \\
&\underset{}{=} \cont^X_A \tensor_X \id_{(X,A)}; [ (\sigma^X_{ (X,\oc A), (X, \oc A \tensor A)}; \partial^X_A  \tensor_X  \id_{(X, \oc A)}) + (\id_{(X, \oc A)} \tensor_X \partial^X_A)] \\
&\underset{}{=} \cont^X_A \tensor_X \id_{(X,A)}; [ ( \id_{(X,\oc A)} \tensor_X \sigma^X_{ (X,\oc A), (X, A)}; \partial^X_A  \tensor_X  \id_{(X, \oc A)}) + (\id_{(X, \oc A)} \tensor_X \partial^X_A)] \qedhere
\end{align*}
\end{proof}

We showed that for every $X \in \C$, $\partial^X$ satisfies the linear rule, the chain rule, and the Leibniz rule. Thus, we know \cite{blute_differential_2020}[cor. 6.11] that $\partial^X$ is a deriving transform. We now state the main theorem.

\begin{thm}
\label{theo_1}
In any generalised differential Seely category, every slice of the linear simple fibration is a differential Seely category.
\end{thm}

\begin{cor}
As an immediate corollary, we have that the monoidal category of a generalised differential Seely category is a differential Seely category.
\end{cor}

\section{Correspondance with Cartesian Differential Categories}
\label{sec:CDC}

Cartesian Differential Categories (CDC) are an axiomatization of differentiation focusing on the cartesian structure instead of the structure of linear maps \cite{blute_cartesian_2009} \cite{cruttwell_cartesian_2017}. They are involved in a correspondence with DSC:
\begin{itemize}
\item The Kleisli category of a DSC is a CDC \cite{blute_cartesian_2009}.
\item Every CDC embeds fully faithfully into the Kleisli category of a DSC \cite{garner_cartesian_2021}.
\end{itemize}

Generalised Cartesian Differential Categories (GCDC) are a generalisation of CDC \cite{cruttwell_cartesian_2017}. It is natural to ask whether or not such a correspondence still holds between GCDC and our GDSC. In this section, we briefly investigate this question. First, we define Cartesian differential categories.

\begin{defi}[Left Additive Categories \cite{blute_cartesian_2009}]
A category $\mathbb{C}$ is left additive if each hom-set is a commutative monoid and $f (g + h) = (f g) + (f h)$ and $f 0 = 0$. 
\end{defi}

\begin{defi}[Cartesian Left Additive Categories \cite{blute_cartesian_2009}]
A Cartesian left additive category is a left additive category with products such that the structure maps $\pi_1$, $\pi_2$, and $\Delta$ are additive and that whenever $f$ and $g$ are additive, then $f \times g$ is additive.
\end{defi}

\begin{defi}[Cartesian Differential Categories \cite{blute_cartesian_2009}]
A CDC is a cartesian left additive category, equipped with an operator $D_\times$ sending a morphism $f : X \to Y$ to a morphism $D_\times [f]:  X \times X \to Y$ such that
\begin{enumerate}[{[CDC.1]}]
\item $D_\times [f+g] = D_\times [f] + D_\times [g]$ and $D_\times[0] = 0$  
\item $\langle v, h + k \rangle; D_\times [f] = \langle v,h \rangle D_\times [f] + \langle v,k \rangle D_\times [f]$ and $\langle v,0 \rangle D_\times [f] = 0$
\item $D_\times [1] = \pi_2$, $ D_\times [\pi_2] = \pi_2;\pi_2$ and $D_\times [\pi_1] = \pi_2;\pi_1$
\item  $D_\times[\langle f, g \rangle] = \langle D_\times [f], D_\times[g] \rangle$
\item $D_\times[f;g] = \langle \pi_1; f, D_\times[f] \rangle; D_\times[g]$ 
\item $ \langle \langle g, h \rangle, \langle 0, k \rangle \rangle; D_\times[D_\times[f]] = \langle g, k \rangle; D_\times[f]$
\item $ \langle \langle g, h \rangle, \langle k, 0 \rangle \rangle; D_\times[D_\times[f]] = \langle \langle g, k \rangle, \langle h, 0 \rangle \rangle; D_\times[D_\times[f]] $
\end{enumerate}
\end{defi}

\begin{exa}
An example of CDC is the category of finite dimentional reel vector spaces with smooth maps. $D_\times$ is given by the Jacobian. 
\end{exa}

Given a DSC $\L$, its Kleisli category $\L_\oc$ is a cartesian left additive category. For every $f : \oc A \to B$, defining $D_\times[f]$ as 
\begin{center}
\begin{tikzcd}[column sep=4em, row sep=4em]
\oc (A \oplus A) \arrow[r, "{s_{A,B}}"] & \oc A \tensor \oc A \arrow[r, "\id_{\oc A} \tensor \dere_A"] & \oc A \tensor A \arrow[r, "\partial_A; f"] & B
\end{tikzcd}
\end{center}
Defines a cartesian differential operator, endowing $\L_\oc$ with a CDC structure. Notice that defining $D_\times[f]$ this way is the same as defining it as $D_\times[f] = \banana{\Tan(f)}; \pi_2$ (see lem. \ref{lem:a1}). 
\begin{center}
\begin{tikzcd}[column sep=4em, row sep=4em]
\Simp(\L_\oc) \arrow[dr, "\simp"'] \arrow[rr, "\F^\Simp"{name=U}, bend left=8, , shift left=0.5]  &  & \LSimp(\L_\oc)  \arrow[ll, "\U^\Simp"{name=D}, bend left=8, shift left=0.5]  \arrow[dl, "\lsimp", shift left= 1] \arrow[phantom, from=U, to=D, "\dashv" rotate=-90]  \arrow[dl, "{\banana{-}}",  shift left= 2, bend left=30]&  \\
& \L_\oc \arrow[ur, "\Tan",shift left=1] & 
\end{tikzcd}
\end{center}
In this case $\banana{\Tan(-)} : \L_\oc \to \L_\oc$ is the tangent structure associated to the CDC $\L_\oc$ \cite{cockett_differential_2014}. This observation raises the question: is it in general the case that in any GDSC, the endofunctor $\banana{\Tan(-)}: \C \to \C$ defines a GCDC? To answer it, we first define generalised Cartesian differential categories.

\begin{defi}[Generalised Cartesian Differential Categories \cite{garner_cartesian_2021}]
A generalised Cartesian differential category consists of a Cartesian category $\mathbb{C}$ with:
\begin{itemize}
\item for each object $X$, a commutative monoid $L(X) = (\lambda(X), +_X , 0_X )$, satisfying  $L(\lambda(X)) = L(X)$ and $L(X \times Y ) = L(X) \times L(Y )$. 
\item for each map $f : X  \to Y$ , a map $D_\times[f] : X \times \lambda(X) \to \lambda(Y)$ that satisfies the following axioms.
\end{itemize} 
\begin{enumerate}[{[GCDC.1]}]
\item $D_\times [+_X] = \pi_2; +_X$ and $D_\times[0_X] = \pi_2;0_X$  
\item $\langle v, h + k \rangle; D_\times [f] = \langle v,h \rangle D_\times [f] + \langle v,k \rangle D_\times [f]$ and $\langle v,0 \rangle D_\times [f] = 0$
\item $D_\times [1] = \pi_2$, $ D_\times [\pi_2] = \pi_2;\pi_2$ and $D_\times [\pi_1] = \pi_2;\pi_1$
\item  $D_\times[\langle f, g \rangle] = \langle D_\times [f], D_\times[g] \rangle$
\item $D_\times[f;g] = \langle \pi_1; f, D_\times[f] \rangle; D_\times[g]$ 
\item $ \langle \langle g, h \rangle, \langle 0, k \rangle \rangle; D_\times[D_\times[f]] = \langle g, k \rangle; D_\times[f]$
\item $ \langle \langle g, h \rangle, \langle k, 0 \rangle \rangle; D_\times[D_\times[f]] = \langle \langle g, k \rangle, \langle h, 0 \rangle \rangle; D_\times[D_\times[f]] $
\end{enumerate}
\end{defi}

\begin{exa}
Every CDC is a GCDC with $\lambda(X) = X$ for every $X$.
\end{exa}

In every CDSC we can define an operator $D_\times$ as $D_\times[f] = \banana{\Tan(f)}; \pi_2$. We know that considering the functor $\U^\Simp \comp \Tan : \C \to \Simp(\C)$ allows us to deduce that this operator satisfies axioms GCDC.1 to GCDC.5 \cite{capucci_fibrational_2024}. The question is whether or not GCDC.6 and GCDC.7 hold. Showing GCDC.6 is straightforward, as it is in essence the same as axiom t.3. Using the same technique used in proving that the Kleisli of DSC is a CDC \cite{blute_cartesian_2009}, we can show that GCDC.7 holds for all morphisms of shape $X \to \U(A)$. This is because such morphisms are representable in $\L$ as morphisms of type $\F(X) \to A$, and we can describe their differential in $\L$ using the generalised deriving transform (see rem. \ref{rem:genpartial}). Proving it for all morphisms is still an open question. Perhaps, one should add one more axiom to the definition of GDSC in order for $\C$ to always be a GCDC.

\section{Conclusion and Future Works}
\label{conclu}

Generalised differential Seely categories are a more general notion of logical differentiation than differential Seely categories. They are fundamentally a dependent construction; limiting oneself to only the fibred adjunction we described is too restrictive. The natural setting for this construction seems to be a general model of dependent LL \cite{pitts_categorical_2015} \cite{lundfall_models_nodate}, where one allows pairs of fibred linear-non-linear adjunctions that are not necessarily $\Simp(C)$ and $\LSimp(\C)$. Moreover, most of the proofs done here could be done in the general setting of comprehension categories, and the fibrational setting for tangent categories described in the recent work of M. Capucci, G. Cruttwell, G. Neil, and F. Zanasi \cite{capucci_fibrational_2024}. \\
This work is the starting point for a search for structures satisfying the more general setting we described. A source of such structures should be cartesian tangent categories, whose differential bundle fibration makes it into a model of dependent LL. Looking for such structures in the tangent category of convenient manifold or in the tangent categories arising as the Eilenberg-Moore category of DSC \cite{cockett_tangent_2020} could be a first step in understanding the mix between dependence and differentiation in LL. Considering tangent categories is natural, as it is the natural dependent generalisation of GCDC. \\

\begin{center}
\begin{tikzcd}[column sep=4em, row sep=3em]
|[draw, rectangle]|{\text{Tangent Categories}} 
  & |[draw, rectangle]|{\text{Dependent DiLL}} \arrow[l, "dual ? "', dotted] \\
|[draw, rectangle]|{\text{GCDC}}  \arrow[u, "generalisation"]
  & |[draw, rectangle]|{\text{GSCD}} \arrow[l, "dual"', dotted]   \arrow[u, "generalisation"', dotted]\\
|[draw, rectangle]|{\text{CDC}} \arrow[u, "generalisation"]
  & |[draw, rectangle]|{\text{DSC}} \arrow[l, "coKleisli"'] \arrow[u, "generalisation"']
\end{tikzcd}
\end{center}

\clearpage

\bibliographystyle{alphaurl}
\bibliography{biblio.bib}

\end{document}